\documentclass[journal]{IEEEtran}
\usepackage{amssymb}
\usepackage{amsmath} 
\usepackage{amsthm}
\usepackage{xcolor}
\usepackage{cite}
\usepackage{booktabs}
\usepackage{caption} 
\usepackage{graphicx} 
\usepackage{threeparttable}
\usepackage{subcaption} 
\usepackage{hyperref}
\theoremstyle{remark}
\newtheorem{remark}{Remark}
\theoremstyle{definition}
\newtheorem{assumption}{Assumption}
\newtheorem{theorem}{Theorem}
\usepackage{tabularx}
\newcolumntype{L}{>{\raggedright\arraybackslash}X}
\usepackage{array, threeparttable, tabularx, booktabs, caption}

\begin{document}
\title{Hierarchical Decision-Making under Uncertainty: A Hybrid MDP and Chance-Constrained MPC Approach}

\author{Siyuan~Li,~\IEEEmembership{Member,~IEEE,}
        Chengyuan~Liu*,~\IEEEmembership{Senior Member,~IEEE,}
        and~Wen-hua~Chen,~\IEEEmembership{Fellow,~IEEE}
\thanks{Siyuan Li and Chengyuan Liu are with the Department of Aeronautical and Automotive Engineering, Loughborough University, Loughborough, UK. (e-mail: s.li8@lboro.ac.uk, c.liu10@lboro.ac.uk).}
\thanks{Wen-hua Chen is with the Department of Aeronautical and Aviation Engineering, The Hong Kong Polytechnic University, Hong Kong. ( e-mail: wenhua.chen@polyu.edu.hk).}
}

\markboth{Journal of \LaTeX\ Class Files,~Vol.~14, No.~8, Dec~2025}%
{Shell \MakeLowercase{\textit{et al.}}: Bare Demo of IEEEtran.cls for IEEE Journals}
\maketitle

\begin{abstract}
This paper presents a hierarchical decision-making framework for autonomous systems operating under uncertainty, demonstrated through autonomous driving as a representative application.
Surrounding agents are modeled using Hybrid Markov Decision Processes (HMDPs) that jointly capture maneuver-level and dynamic-level uncertainties, enabling the multi-modal environmental prediction.
The ego agent is modeled using a separate HMDP and integrated into a Model Predictive Control (MPC) framework that unifies maneuver selection with dynamic feasibility within a single optimization.
A set of joint chance constraints serves as the bridge between environmental prediction and optimization, incorporating multi-modal environment predictions into the MPC formulation and ensuring safety across all plausible interaction scenarios. 
The proposed framework provides theoretical guarantees on recursive feasibility and asymptotic stability, and its benefits in terms of safety and efficiency are validated through comprehensive evaluations in highway and urban environments, together with comparisons against a rule-based baseline.
\end{abstract}


\IEEEpeerreviewmaketitle

\section{Introduction}\label{I}

Autonomous systems operate in uncertain and dynamically evolving environments shared with humans, robots, and other intelligent agents~\cite{yang2025recent}. To make reliable decisions under such uncertainty, they must not only perceive the scene accurately but also reason about the behavioral uncertainty of neighboring agents, which manifests at two coupled levels: the maneuver level (e.g., lane-change intent, route choice, or task assignment) and the dynamic level (e.g., velocity, heading, and pose). Capturing and predicting these uncertainties are essential for anticipating how surrounding agents (SAs) will interact with the ego system. This necessitates robust environment modeling, which forms the foundation for reliable decision-making under uncertainty.

Recent studies on environment modeling have made notable progress along two complementary research lines, focusing respectively on maneuver-level and dynamic-level uncertainty.
Regarding  maneuver-level uncertainty, the modeling approaches have evolved from early physics-based motion models\cite{uhlemann2024constantvelocity,ivanchev2019constantacc}, which fail to capture strategic maneuver shifts, to more expressive maneuver reasoning frameworks. Logic-based methods explicitly treat interacting agents as rational decision-makers that optimize certain objectives, enabling intuitive and interpretable reasoning about their high-level intentions~\cite{bai2015intention,chen2025intention}. In parallel, discriminative deep models, such as LSTMs~\cite{wang2021intelligent}, Transformers~\cite{li2024behavior}, or GNNs~\cite{zhang2023bip}, learn maneuver patterns from motion sequences and interaction cues for intent classification. While these methods offer high accuracy in identifying maneuver modes, their outputs are often decoupled from downstream trajectory generation, lacking structured integration.

For dynamic-level uncertainty, classical approaches in control theory are typically grounded in stochastic system theory. For example, Kalman filters and their nonlinear extensions provide real-time state estimates and the associated uncertainty covariance \cite{carvalho2014stochastic,xu2014motion}; reachability analysis offers formal boundary guarantees for dynamics prediction, especially in safety-critical verification scenarios~\cite{xu2022reachability,shetty2019predicting}. In recent years, the emergence of deep generative models has fueled the development of data-driven approaches. Representative architectures such as LSTMs~\cite{li2022combining}, Transformers~\cite{ma2024multiagent}, diffusion models~\cite{wang2024trajectory}, and CVAEs~\cite{ivanovic2020multimodal} generate multi-modal and nonlinear distributions over future dynamics, capturing complex motion patterns and interaction features. However, both control-based and learning-based models often rely on predefined or independently inferred  maneuvers as a prerequisite. They typically lack an integrated framework that connects maneuver reasoning with dynamics forecasting, and thus fail to explain the decision-making motivations underlying agent maneuver.

Overall, most existing studies treat maneuver and dynamics forecasting as decoupled subproblems, limiting the decision-making system's ability to reason about joint risk in a coherent manner. Therefore, developing a structured prediction framework that jointly models both maneuver-level and dynamic-level uncertainty, while maintaining interpretability and planning compatibility, is a critical step toward reliable environment understanding.

While accurate environmental prediction is essential, safe and goal-directed autonomy ultimately depends on how these environmental predictions are incorporated into the decision-making process. Existing decision-making paradigms can be broadly categorized into rule-based, learning-based, and model-based approaches. Rule-based strategies~\cite{pickering2025narrative,xu2022integrated,kang2025decision} are widely adopted for their simplicity and interpretability, but they rely heavily on manually designed heuristics, making them rigid and reactive in uncertain environments. Learning-based methods, such as reinforcement learning~\cite{zhang2025decision,wu2025reinforcement} and imitation learning~\cite{andreychuk2025mapf,yang2025optimization}, can achieve strong empirical performance but often suffer from limited interpretability and lack of safety guarantees. Model-based frameworks, particularly those based on Model Predictive Control (MPC)~\cite{wang2022expert,caregnato2023real,yang2025interactive}, provide a principled optimization structure that explicitly accounts for system dynamics and safety constraints, offering better transparency and adaptability. 

Despite their advantages, existing MPC-based frameworks still exhibit structural limitations in representing the ego agent (EA)’s decision-making process.
Specifically, most existing MPC-based decision frameworks treat the EA’s maneuver selection and motion planning as separate processes, without explicitly constraining high-level maneuver decisions by the agent's dynamic feasibility. Moreover, environment-derived predictions are rarely incorporated in an explicit form, limiting the planner’s ability to reason about future interaction uncertainties. Therefore, developing a unified decision-making framework that couples the EA’s maneuver reasoning and dynamic feasibility with environmental predictions within a MPC framework remains an open challenge.

To tackle these challenges, this paper proposes a hierarchical decision-making framework that combines Hybrid Markov Decision Processes (HMDPs)\cite{wang2025high} with MPC. In this framework, each SA is represented as an HMDP that explicitly distinguishes and jointly models maneuver-level and dynamic-level uncertainties. The maneuver-level uncertainty is captured by a probabilistic action-selection distribution that encodes the likelihood of discrete maneuver transitions, whereas the dynamic-level uncertainty is represented by Gaussian process noise added to the state evolution, thereby accounting for stochastic deviations around nominal motion. To further mitigate maneuver uncertainty, a likelihood-based filtering mechanism suppresses low-probability transitions over both single-step and multi-step horizons, yielding compact multi-modal predictions that retain only the most plausible interaction outcomes, which are subsequently propagated to the EA’s decision layer as inputs for optimization.

Building upon the structured environment representation, the EA’s decision-making process is formulated in a hybrid manner consistent with that of the environment model. Unlike the generative HMDPs used for SAs, its HMDP serves as a decision structure. Discrete maneuver options are treated as decision variables, each explicitly linked to the EA’s dynamic model. Embedding this hybrid structure into the MPC formulation couples maneuver selection with continuous dynamics, allowing the dynamic feasibility of each maneuver to be explicitly evaluated and enforced through the optimization constraints. This integration eliminates inconsistency between high-level maneuver planning and low-level motion control, yielding decisions that are both dynamically feasible and behaviorally consistent.

The multi-modal predictions generated by the environmental HMDPs are incorporated into the MPC optimization via joint chance constraints, which link environment prediction with decision-making and enforce probabilistic safety across all retained scenarios. Under standard affine-Gaussian assumptions, these chance constraints admit a deterministic reformulation, allowing environment-derived safety margins to be efficiently enforced. Overall, the proposed framework unifies environment prediction and decision-making within a coherent optimization framework, producing decisions that are both forward-looking and adaptive to the stochastic evolution of the surrounding environment.

The main contributions of this paper are as follows:
\begin{itemize}
    \item An HMDP-based environment modeling framework is developed for SAs, which jointly captures maneuver-level and dynamic-level uncertainties to generate structured, multi-modal  predictions of their future behaviors. 
    \item  A unified HMDP-MPC framework is developed for the EA, integrating maneuver selection with dynamic feasibility and embedding multi-modal environment predictions through joint chance constraints for probabilistically safe planning.
    \item Theoretical guarantees on recursive feasibility and asymptotic stability are established for the proposed decision-making framework. 
    \item Comprehensive simulations in both highway and urban autonomous driving scenarios are performed to validate the effectiveness of the proposed framework and to demonstrate its key differences from a rule-based baseline.
\end{itemize}

The rest of the paper is organized as follows:
Section~\ref{II} introduces the proposed hierarchical decision-making framework. Section~\ref{III} provides theoretical proofs of recursive feasibility and asymptotic stability. Section~\ref{IV} presents simulation results in the context of an autonomous driving scenario. Section~\ref{V} offers a discussion of the results and outlines potential directions for future research.

\section{Hierarchical Decision-Making Framework}\label{II}

Building upon the modeling principles introduced in Section~\ref{I}, this section formulates the proposed hierarchical decision-making framework. As illustrated in Fig.~\ref{fig:framework}, data from the Perception Layer provide the current states of the EA and all SAs. These states form the initial conditions for both the environment side and the decision-making layer. On the ego side, the EA HMDP Modelling block (Section~\ref{B}) specifies the EA’s discrete maneuver modes and action-dependent dynamics for use in the MPC optimization. On the environment side, the SAs Modelling \& Prediction block uses HMDP models (Section~\ref{A}) to propagate these states and employs probabilistic filtering (Section~\ref{C}, Section~\ref{D}) to prune low-likelihood maneuver branches, yielding multi-step reachability sets that describe all plausible future evolutions of the SAs. These reachability sets are then passed to the Optimization Layer, where they define joint chance constraints that restrict the EA’s feasible decision space (Section~\ref{D}). The MPC solver (Section~\ref{E}) computes an action sequence that is dynamically feasible and probabilistically safe, and its first control input is sent to the Low-Level Layer.

\begin{figure}[htbp]
    \captionsetup{font=footnotesize}
    \centering
    \includegraphics[width=0.49\textwidth]{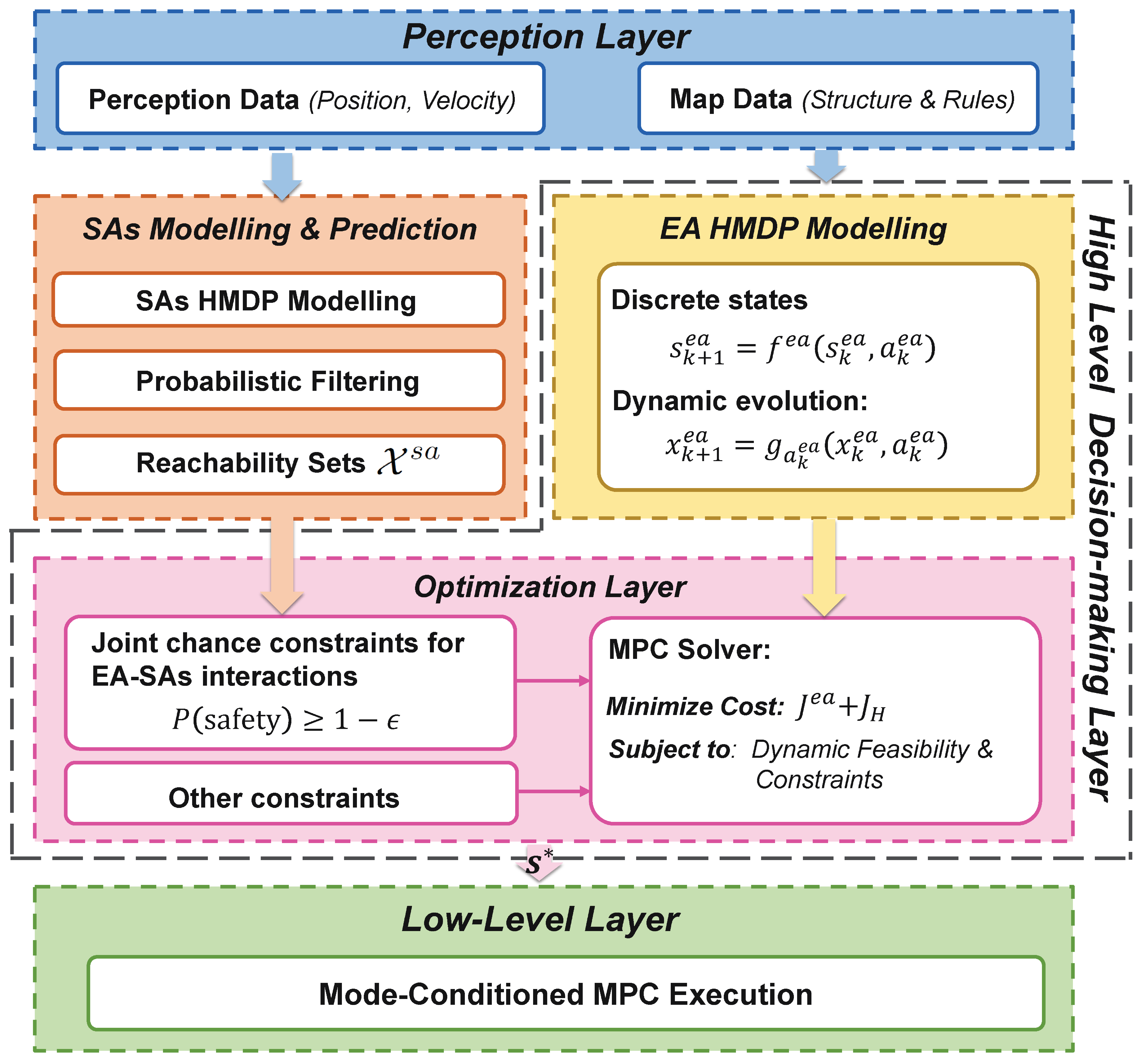}
    \caption{Schematic overview of the proposed hierarchical decision-making framework.}
    \label{fig:framework}
\end{figure}

\subsection{HMDP Formulation for the SAs}\label{A}

Unlike our previous work~\cite{wang2025high,li2024integrated}, which assumed fully observable environment dynamics, the present formulation explicitly incorporates maneuvers-level and  dynamic-level uncertainty through independent HMDPs for SAs. 

 A 3-tuple MDP, denoted \((S^{sa}, A^{sa}, \pi^{sa})\), is employed to model the maneuver evolution of SAs. The set \(S^{sa}\) contains discrete states, which represent high-level or abstract maneuvers, while \(A^{sa}\) denotes the set of actions that induce transitions between these states. The policy \(\pi^{sa}\) specifies a probability distribution over the action space \(A^{sa}\) conditioned on the current state \(s^{sa}\). More specifically, for the \(i\)-th SA at time step \(k\), the probability of selecting action \(a_k^{sa,i} \in A^{sa,i}\) given the current state \(s_k^{sa,i} \in S^{sa,i}\) is defined as:
\begin{equation*}
    \mathbb{P}(a_k^{sa,i} \mid s_k^{sa,i}) = \pi^{sa,i}(a_k^{sa,i} \mid s_k^{sa,i}).
\end{equation*}
The maneuver-level uncertainty is captured by the stochastic policy
$\pi^{sa,i}(a_k^{sa,i}\mid s_k^{sa,i})$, which assigns probabilities to feasible maneuvers. However, due to physical or structural constraints, some actions are rendered infeasible in specific states, resulting in $\mathbb{P}(a_k^{sa,i} \mid s_k^{sa,i}) = 0$.
This motivates defining a state-dependent admissible action set, restricting the policy to actions that are feasible in a given state. Formally, for the \(i\)-th SA, we have
\begin{equation*}\label{eq:Asetphysical}
    A^{sa}(s_k^{sa,i}) \subseteq A^{sa},
\end{equation*}
where \(A^{sa}(s_k^{sa,i})\) denotes the set of actions feasible at state \(s_k^{sa,i}\), excluding transitions that violate system-level or physical feasibility conditions, thereby ensuring that the resulting policy remains realizable.

Once a feasible action is selected according to the policy, the next discrete state evolves according to a deterministic transition function:
\begin{equation*}\label{eq:sastatetransition}
    s_{k+1}^{sa,i} = f^{sa}(s_k^{sa,i}, a_k^{sa,i}), \quad s_{k+1}^{sa,i}\in S^{sa,i}
\end{equation*}
which captures predefined transition rules governing the discrete-mode evolution.

\begin{remark}
This formulation can be extended to stochastic transitions by introducing uncertainty into the transition function \(f^{sa}\), or by modeling the next state as a distribution conditioned on the current state-action pair:
\begin{equation*}
s_{k+1}^{sa,i} \sim \mathbb{P}(\cdot \mid s_k^{sa,i}, a_k^{sa,i}).
\end{equation*}
In other words, although transitions are designed to follow predefined rules, stochasticity may arise from environmental variability, modeling approximations, or external disturbances, leading to deviations from the deterministic maneuver. For clarity, this work adopts deterministic mode transitions, however, the same framework can accommodate stochastic transitions by replacing $f^{sa}$ with a transition kernel $\mathbb{P}(\cdot|s,a)$.
\end{remark}

Building upon the MDP formulation of discrete maneuver modes, we further specify a set of action-dependent dynamics, enabling the framework to explicitly capture how each high-level decision drives the dynamic evolution of the agent's physical states. The evolution of the \(i\)-th SA can then be represented by the following affine dynamic model
\begin{equation}\label{eq:sudy}
    x_{k+1}^{sa,i} = \mathsf{A}_{k}^{sa,i}(a_k^{sa,i}) x_k^{sa,i} + \mathsf{B}_{k}^{sa,i}(a_k^{sa,i}) \mathsf{F}_k^i(a_k^{sa,i}) + w_k^i,
\end{equation}
where \(x_k^{sa,i} \in \mathbb{R}^n\) denotes the state of the \(i\)-th SA at time step \(k\), \(\mathsf{A}_k^{sa,i}(\cdot)\) and \(\mathsf{B}_k^{sa,i}(\cdot)\) are the system matrices associated with the selected action \(a_k^{sa,i}\), \(\mathsf{F}_k^i(\cdot)\) denotes a hypothetical control input associated with the selected action \(a_k^{sa,i}\), and $w_k^i$ represents stochastic modeling uncertainty arising from unmodeled dynamics and approximation errors. 
It is modeled as zero-mean Gaussian noise, $w_k^i \sim \mathcal{N}(0, \Xi_k^{sa,i}(a_k^{sa,i}))$, 
and is assumed independent across agents and time.

\begin{assumption}
We assume that the matrices \(\mathsf{A}_{k}^{sa,i}(\cdot)\) and \(\mathsf{B}_{k}^{sa,i}(\cdot)\), as well as the mode-dependent input \(\mathsf{F}_k^i(\cdot)\), are known. 
\end{assumption}
It is worth to note that the input \(\mathsf{F}_k^i(\cdot)\) is not the actual control input applied by the SAs, but a hypothetical proxy signal inferred from high-level maneuver features associated with action \(a_k^{sa,i}\). 
\begin{remark}
The formulation in \eqref{eq:sudy} adopts an action-dependent representation of system dynamics, where the matrices \(\mathsf{A}_k^{sa,i}(\cdot)\), \(\mathsf{B}_k^{sa,i}(\cdot)\), and \(\mathsf{F}_k^i(\cdot)\) are parameterized by the discrete action \(a_k^{sa,i}\). However, in certain scenarios, the system dynamics may primarily depend on the state. In such cases, these matrices can alternatively be expressed as functions of the state \(s_k^{sa,i}\), or jointly as functions of both the state and the action. This modeling choice depends on the specific application context and is not strictly deterministic.
\end{remark}

\begin{remark}
We adopt a linear affine function to describe the action-dependent dynamics of SAs for modeling simplicity without compromising generality. However, nonlinear dynamics may arise in practical cases, for example:
\begin{equation*}
    x_{k+1}^{sa,i} = g_{a_k^{sa,i}}(x_k^{sa,i}, \mathsf{F}_k^i(a_k^{sa,i})) + w_k^i,
\end{equation*}
where \(g_{a_k^{sa,i}}(\cdot)\) is a nonlinear function associated with the action \(a_k^{sa,i}\), which can be approximated using linearization methods. In particular, techniques such as Taylor series expansion around nominal operating points are commonly used to obtain linearized models. Therefore, adopting a linear affine representation is both reasonable and accepted in practice.
\end{remark}

Due to the complexity of the environment, it is necessary to consider all SAs that may influence the EA. To better characterize the interaction between the EA and all relevant SAs, we define the aggregated external state at time \(k+1\) as:
\begin{equation*}
Z_{k+1}^{sa} := \left\{ s_{k+1}^{sa,i}, x_{k+1}^{sa,i} \right\}_{i=1}^{I},
\end{equation*}
where \(I\) denotes the number of SAs  included in the environment model. This aggregated representation \(Z_{k+1}^{sa}\) provides the necessary information for the EA to evaluate interaction constraints and determine the set of feasible actions.

\subsection{HMDP Formulation for the EA}\label{B}

In contrast to the modeling of SAs, the EA integrates both its internal state information and the aggregated external state \(Z^{sa}\) (see Section~\ref{A}), so that the policy simultaneously considers ego performance and interaction safety. 
The EA's HMDP is formulated as:
\begin{equation*}
 \left(S^{ea}, A^{ea}, X^{ea}, Z^{sa}, J^{ea}\right),
\end{equation*}
where \(S^{ea}\) denotes the set of ego maneuvers, \(A^{ea}\) is the set of candidate actions, \(X^{ea} \subset \mathbb{R}^n\) represents the EA's dynamic states space, \(Z^{sa}\) captures the state information of relevant SAs, and \(J^{ea}\) is a cost function used to evaluate each action.

As discussed earlier, not all actions are admissible in every discrete mode \( s_k^{ea} \in S^{ea} \) at time step \(k\). Therefore, each state \( s_k^{ea} \) is associated with a subset of feasible actions:
\begin{equation}\label{eq:eaactionset}
A^{ea}(s_k^{ea}) \subseteq A^{ea}.
\end{equation}

Given a feasible action \( a_k^{ea} \in A^{ea}(s_k^{ea}) \), the EA transitions to the next discrete state according to a deterministic transition function:
\begin{equation}\label{eq:satransition}
s_{k+1}^{ea} = f^{ea}(s_k^{ea}, a_k^{ea}), \quad s_{k+1}^{ea}\in S^{ea}.
\end{equation}

Meanwhile, the EA's dynamic state evolves according to action-dependent dynamics:
\begin{equation}\label{eq:eadynamic}
x_{k+1}^{ea} = g_{a_k^{ea}}(x_k^{ea}, a_k^{ea}),
\end{equation}

where \( g_{a_k^{ea}}(\cdot) \) denotes the dynamic function conditioned on the discrete action \( a_k^{ea} \). For instance, if \( a_k^{ea} \) corresponds to a left lane change for autonomous vehicles, then \( g_{a_k^{ea}}(\cdot) \) characterizes lateral polynomial trajectory generation. This action-conditioned dynamic formulation establishes a direct mapping between maneuver selection and dynamic evolution, enabling the dynamic feasibility of each maneuver to be evaluated and enforced within the optimization framework.

In addition to internal constraints, the EA’s feasible actions are further restricted by external conditions derived from the aggregated external state \(Z_{k+1|k}^{sa}\). For each candidate action, the evolution of the EA’s dynamic state can be predicted via~\eqref{eq:eadynamic}, enabling the assessment of potential interactions with SAs. The externally constrained feasible action set at state \(s_k^{ea}\) is therefore defined as
\begin{equation*}
A_{\text{fea}}^{ea}(s_k^{ea}, Z_{k+1|k}^{sa}) \subseteq A^{ea}(s_k^{ea}) \subseteq A^{ea},
\end{equation*}
where \(Z_{k+1|k}^{sa}\) denotes the one-step-ahead predicted states of SAs.
\begin{remark}
The external constraints considered here may arise from safety margins imposed by the environment. For instance, for autonomous driving, if a SV is predicted to occupy a neighboring lane in the near future, a lane change action by the EA may lead to a collision risk, rendering it infeasible. In this case, such an action would be excluded from the externally constrained feasible action set \( A_{\text{fea}}^{ea}(s_k^{ea}, Z_{k+1|k}^{sa}) \). The mechanism for incorporating environment information into action feasibility assessment will be detailed in the next section.
\end{remark}

\subsection{Probabilistic Interaction Modeling with Chance Constraints}\label{C}

For SAs, it is essential to predict the possible maneuver ranges associated with each action \(a_k^{sa,i} \in A^{sa}(s_k^{sa,i})\), so as to capture diverse maneuver possibilities and facilitate informed decision-making. However, in certain cases, the action space \(A^{sa}(s_k^{sa,i})\) is considerably large, thereby increasing the number of state-action sequences to be evaluated and resulting in substantial computational demands.

To address this, we propose an action filtering mechanism based on policy distribution. In practice, many actions within the admissible set have negligible probability. This filtering mechanism focuses on plausible and high-likelihood interactions by discarding low-probability actions. More specifically, we define a threshold \( \delta \), and retain only those actions \( a_k^{sa,i} \in A^{sa}(s_k^{sa,i}) \) whose policy probabilities exceed \( \delta \). The resulting filtered action set is defined as:
\begin{equation*}
\tilde{A}^{sa}(s_k^{sa,i}) = \left\{ a \in A^{sa}(s_k^{sa,i}) \ \middle| \ \pi^{sa,i}(s_k^{sa,i}, a) \geq \delta \right\}.
\end{equation*}
 This naturally induces the following hierarchical relation:
\begin{equation*}
\tilde{A}^{sa}(s_k^{sa,i}) \subseteq A^{sa}(s_k^{sa,i}) \subseteq A^{sa}.
\end{equation*}

For each retained high-likelihood action \(a_k^{sa,i} \in \tilde{A}^{sa}(s_k^{sa,i})\), we propagate the dynamics forward to obtain the next state of the \(i\)-th SA:
\begin{equation}\label{eq:sadynamic}
\tilde{x}_{k+1}^{sa,i} = \mathsf{A}_{k}^{sa,i}(a_k^{sa,i}) \tilde{x}_k^{sa,i} + \mathsf{B}_{k}^{sa,i}(a_k^{sa,i}) \mathsf{F}_k^i(a_k^{sa,i}) + w_k^i.
\end{equation}

Accordingly, the set of all reachable surrounding dynamic states at time step \(k+1\) can be constructed as:
\begin{equation}\label{eq:sareachset}
\mathcal{X}_{k+1}^{sa,i} := \textstyle \bigcup_{a_{k}^{sa,i} \in \tilde{A}^{sa}(s_k^{sa,i})} \tilde{x}_{k+1}^{sa,i}.
\end{equation}

Then, the filtered external information is incorporated into the decision-making process to eliminate actions that may violate safety or other constraints. 

Since the future states of SAs are subject to inherent uncertainty and model mismatch represented by \( w_k^i \), we incorporate chance constraints into the EA’s decision process. This prevents overly conservative maneuver while still ensuring safety and other required guarantees.

Moreover, as each SA may execute multiple filtered admissible actions, the EA must ensure safety under all plausible future evolutions to guarantee robustness. Accordingly, for the \(i\)-th SAs, the feasible ego action set at time step \(k+1\) must satisfy the following chance constraint:
\begin{equation*}
\mathbb{P}\textstyle
\left(  \bigcap_{x_{k+1}^{sa,i} \in \mathcal{X}^{sa,i}_{k+1}}
    \left\{ h(x_{k+1}^{ea}, x_{k+1}^{sa,i})\geq 0\right\}\right) \geq 1 - \epsilon,
\end{equation*}

where \(h(\cdot)\) encodes application-specific constraints (e.g., safety margins), and \(\epsilon\) denotes a user-defined tolerance for allowable risk. For tractability, \(h(x_{k+1}^{ea}, x_{k+1}^{sa,i})\) is typically designed to be affine in the stochastic argument \(x_{k+1}^{sa,i}\). When \(h\) is nonlinear, we adopt a first-order Taylor approximation with respect to \(x_{k+1}^{sa,i}\) around a nominal prediction, yielding an affine surrogate on which the chance constraint is enforced.

\begin{remark}
Concretely, for a nominal prediction \(\bar{x}_{k+1}^{sa,i}\),
\begin{align*}
   h(x_{k+1}^{ea}, x_{k+1}^{sa,i})
&\approx 
h(x_{k+1}^{ea}, \bar{x}_{k+1}^{sa,i})\\
&+ \nabla h(x_{k+1}^{ea}, \bar{x}_{k+1}^{sa,i})^\top
\big(x_{k+1}^{sa,i}-\bar{x}_{k+1}^{sa,i}\big), 
\end{align*}
which provides an affine, computationally tractable surrogate. 
\end{remark}

Extending this to the multi-agent case, the EA must ensure joint safety with respect to all SAs:
\begin{align}\label{eq:chanceallsas}
\mathbb{P}\textstyle\left(
    \bigcap_{i=1}^{I}
   \bigcap_{x_{k+1}^{sa,i} \in \mathcal{X}^{sa,i}_{k+1}}
    \left\{ h(x_{k+1}^{ea}, x_{k+1}^{sa,i})\geq 0\right\}\right) \geq 1-\epsilon   
\end{align}

This formulation ensures that the EA remains safe under all likely  interactions, accounting for multi-agent uncertainty in a principled and tractable way.

\subsection{Chance-Constrained MPC Formulation}\label{D}

In practical decision-making scenarios, prediction over an extended horizon is crucial, as it provides richer contextual information and allows the system to anticipate future interactions. MPC is widely adopted in these settings because it optimizes control inputs over a finite prediction horizon while explicitly accounting for constraints. This predictive optimization structure naturally enables the integration of future information into real-time decision-making.

However, extending the prediction horizon inevitably leads to a combinatorial growth in the number of possible maneuver sequences. To address this, we extend the previously proposed filtering mechanism to a multi-step setting, where cumulative probability threshold is used to prune low-likelihood maneuver sequences. 

For each SA \( i \), we define the feasible action sequence set \( \mathcal{A}^{sa,i}_{\delta} \) under a cumulative probability constraint as:
\begin{equation}\label{eq:cumprobability}
\mathcal{A}^{sa,i}_{\delta} = \left\{ a_{1:H}^{sa,i} \,\middle|\, \prod_{k=1}^{H} \pi^{sa,i}(s_k^{sa,i}, a_k^{sa,i}) \geq \delta_{\text{seq}} \right\},
\end{equation}
where \( a_{1:H}^{sa,i} = [a_1^{sa,i}, \dots, a_H^{sa,i}] \) is a candidate action sequence of length \( H \), and \( \delta_{\text{seq}} \) is a scenario-level probability threshold. 
 
For a given action sequence \( a_{1:H}^{sa,i} \), the corresponding discrete state sequence \( \boldsymbol{s}_{2:H+1}^{sa,i}  = [s_2^{sa,i}, \dots, s_{H+1}^{sa,i}]\) is deterministically generated according to equation (\ref{eq:satransition}). The prediction of the SA's system state follows the same affine stochastic dynamics as introduced in (\ref{eq:sadynamic}), which serve as the prediction model for forward propagation:
\begin{equation}\label{eq:saprediction}
    x_{k+1|k}^{sa,i} = \mathsf{A}_k^{sa,i}(a_k^{sa,i})\, x_{k|k}^{sa,i} + \mathsf{B}_k^{sa,i}(a_k^{sa,i})\, \mathsf{F}_k^{i}(a_k^{sa,i}) + w_k^i.
\end{equation}

To quantify the uncertainty associated with this prediction, we also propagate the state covariance through the linearized system. The covariance is updated by:
\begin{equation}\label{eq:coinvpro}
    Q_{k+1|k}^{sa,i} = \mathsf{A}_k^{sa,i}(a_k^{sa,i})\, Q_{k|k}^{sa,i} \left(\mathsf{A}_k^{sa,i}(a_k^{sa,i})\right)^\top+\Xi_k^{sa,i}(a_k^{sa,i}).
\end{equation}
Here, $Q^{sa,i}_{k|k}$ and $Q^{sa,i}_{k+1|k}$ denote the posterior and predicted state covariance of SA~$i$, respectively, while $\Xi^{sa,i}_k(a_k^{sa,i})$ is the process-noise covariance capturing model uncertainty.

For each retained action sequence \( a_{1:H}^{sa,i} \in \mathcal{A}_\delta^{sa,i} \), we denote the corresponding predicted state trajectory as:
\[
\tilde{x}^{sa,i}(a_{1:H}^{sa,i}) = [x_{2|1}^{sa,i}, x_{3|1}^{sa,i}, \dots, x_{H+1|1}^{sa,i}],
\]
which is derived from the dynamics defined in Equation~(\ref{eq:saprediction}). Then, the predicted dynamic state reachability set in Equation~(\ref{eq:sareachset}) becomes:
\begin{equation}\label{eq:seqset}
\mathcal{X}^{sa,i} := \textstyle \bigcup_{a_{1:H}^{sa,i} \in \mathcal{A}_\delta^{sa,i}} \tilde{x}^{sa,i}(a_{1:H}^{sa,i}).
\end{equation}
To ensure safety under multi-step interaction uncertainty, the chance constraint~\eqref{eq:chanceallsas} is extended to:
\begin{equation}\label{eq:mixchanceconstraints}
\mathbb{P}\left(
\bigcap_{i=1}^{I}\bigcap_{j=1}^{H}
\bigcap_{\tilde{x}^{sa,i} \in \mathcal{X}^{sa,i}}
\left\{ h\left(x_{k+j|k}^{ea}, \tilde{x}_{k+j|k
}^{sa,i} \right) \geq 0 \right\}
\right)
\geq  1-\epsilon.
\end{equation}

Equations (\ref{eq:cumprobability})-(\ref{eq:mixchanceconstraints}) provide an abstract description of the probabilistic interaction model and the resulting multi-step
joint chance constraints. To give intuition, Fig.~\ref{fig:example} illustrates
this construction for a single surrounding vehicle (SV) in a three-lane scenario.

\begin{figure}[t]
    \captionsetup{font=footnotesize}
    \centering
    \includegraphics[width=0.49\textwidth]{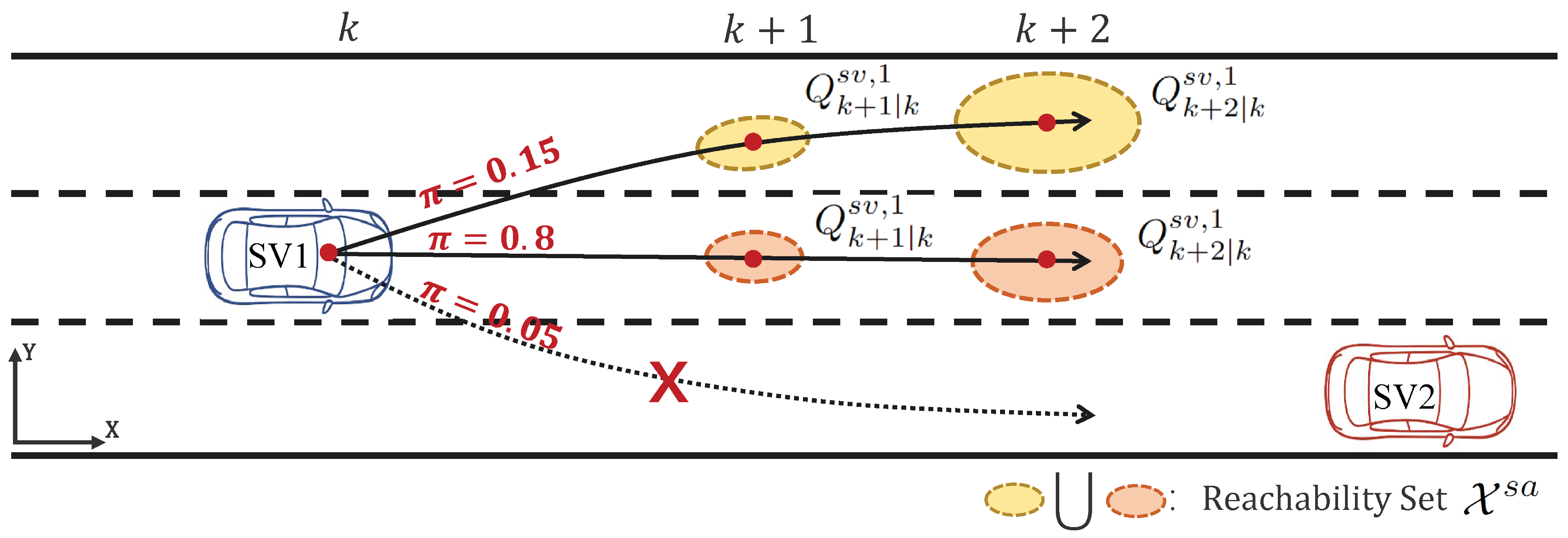}
     \caption{Illustration of multi-modal prediction and reachability-set construction of a SV. Low-probability  branch is
pruned, while the remaining hypotheses are propagated to obtain predicted means, covariances, and ellipsoidal reachable regions whose union approximates the reachability set $\mathcal{X}^{sa}$ used in the joint chance constraint~(\ref{eq:mixchanceconstraints}).}
    \label{fig:example}
\end{figure}

Having specified the probabilistic interaction model, the multi-step reachability sets, and the associated joint chance constraints, we now define the cost function used to select among the feasible actions. This cost function is formulated as
\begin{align*}
J^{ea}: S^{ea} \times A^{ea} \times X^{ea} \times \mathcal{X}^{sa} \rightarrow \mathbb{R}.
\end{align*}

In practical implementations, the specific structure of the cost function \( J^{ea} \) may vary depending on the application. Not all elements of the input tuple need to be explicitly used, the function can be simplified to depend only on the subset of variables relevant to the specific task.

At each time step \(k\), the EA solves a finite-horizon optimization problem over a prediction horizon \(H\), aiming to select a sequence of high-level actions \(\{a_k^{ea}, a_{k+1}^{ea}, \dots, a_{H-1}^{ea}\}\) that minimize the cumulative cost:
\begin{align*}
    \min_{a^{ea}} \sum_{j=0}^{H-1} &J^{ea}(s_{k+j+1|k}^{ea}, a_{k+j|k}^{ea}, x_{k+j+1|k}^{ea}, \mathcal{X}_{k+j+1|k}^{sa}) \\ & +J_H(s^{ea}_{k+H|k}, x^{ea}_{k+H|k})
\end{align*}

\text{s.t.}
\begin{equation}
   (\ref{eq:eaactionset})-(\ref{eq:eadynamic}),and \ (\ref{eq:cumprobability})-
  (\ref{eq:mixchanceconstraints}). \nonumber
\end{equation}

where $J_H(\cdot)$ represents a terminal cost associated with the finite-horizon problem, the explicit definition of which will be introduced in Section~\ref{III}.

\subsection{Reformulation for the MPC Optimization Program}\label{E}

The direct evaluation of chance constraints \eqref{eq:mixchanceconstraints} involves multivariate probability integrals, which are computationally prohibitive for real-time applications. To address this intractability, we reformulate these probabilistic constraints into deterministic algebraic surrogates under standard affine-Gaussian assumptions. This transformation allows safety margins to be efficiently computed from the predicted state means and covariances, thereby enabling the evaluation of candidate discrete decision sequences within the optimization framework.

Specifically, the affine constraint function \( h(\cdot) \) is expressed as:
\begin{equation*}
h(x_{k+1|k}^{ea}, \tilde{x}_{k+1|k}^{sa,i}) = c^{ea} x_{k+1|k}^{ea} + c^{sa} \tilde{x}_{k+1|k}^{sa,i} + c,
\end{equation*}
where \( c^{ea}, c^{sa}, \) and \( c \) are known constants. Moreover, \( \tilde{x}_{k+1|k}^{sa,i} \) is assumed to follow a Gaussian distribution with mean \( \bar{x}_{k+1|k}^{sa,i} \) and covariance \( Q_{k+1|k}^{sa,i} \).

Thus, the chance constraint:
\begin{equation}\label{eq:chance_constraints}
\mathbb{P} \left(
    h(x_{k+1|k}^{ea}, \tilde{x}_{k+1|k}^{sa,i}) \geq 0
\right) \geq 1 - \epsilon
\end{equation}
can be equivalently written as \cite{calafiore2006distributionally}:
\begin{equation*} \label{eq:det_chance}
    c^{ea} x_{k+1|k}^{ea} +c^{sa}\bar{x}_{k+1|k}^{sa,i} + c \geq \Phi^{-1}(1 - \epsilon) \cdot \sqrt{ (c^{sa})^\top Q_{k+1|k}^{sa,i} c^{sa} },
\end{equation*}
where \( \Phi^{-1}(\cdot) \) denotes the inverse cumulative distribution function of the standard Gaussian distribution.

Accordingly, the multi-agent chance constraint~\eqref{eq:mixchanceconstraints}
can be deterministically represented as
\begin{align}\label{eq:det_constraints}
\textstyle\bigcap_{i=1}^{I}\bigcap_{j=1}^{H}
\Big\{
& c^{ea} x_{k+j|k}^{ea} + c^{sa} \bar{x}_{k+j|k}^{sa,i} + c \nonumber \\
& \geq \Phi^{-1}(1 - \epsilon )
\sqrt{(c^{sa})^\top Q_{k+j|k}^{sa,i} c^{sa}}
\Big\},
\end{align}
which must hold for all predicted sequences \(\tilde{\boldsymbol{x}}^{sa,i}\) 
through their induced \(\bar{x}_{k+j|k}^{sa,i}\) and \(Q_{k+j|k}^{sa,i}\).

Finally, the resulting chance-constrained MPC problem is formulated as:
\begin{align}\label{eq:etermpc}
    \min_{a^{ea}} \sum_{j=0}^{H-1} &J^{ea}(s_{k+j+1|k}^{ea}, a_{k+j|k}^{ea}, x_{k+j+1|k}^{ea}, \mathcal{X}_{k+j+1|k}^{sa}) \\ & +J_H(s^{ea}_{k+H|k}, x^{ea}_{k+H|k})\nonumber
\end{align}

\text{s.t.}
\begin{equation}
  (\ref{eq:eaactionset})-(\ref{eq:eadynamic}), \ (\ref{eq:cumprobability})-
  (\ref{eq:seqset}),
  and \ (\ref{eq:det_constraints}). \nonumber
\end{equation}

By solving \eqref{eq:etermpc}, an optimal action sequence is obtained, of which only the first action is applied to the low-level controller. 

\section{ Recursive Feasibility and Asymptotic Stability }\label{III}
Building upon the complete formulation, the following section establishes theoretical guarantees of recursive feasibility and stability, which are fundamental to safe and reliable decision-making.

\begin{assumption}[Goal mode and baseline policy]
\label{baseline}
Let $s^{ea}_{\mathrm{goal}} \in \mathcal{S}^{ea}$ denote a goal mode
of the EA. The stage cost $J^{ea}(s^{ea},x^{ea},a^{ea})$
is nonnegative and can be normalised such that
\begin{equation}
  J^{ea}(s^{ea},a^{ea}, x^{ea},\mathcal{X}^{sa}) = 0
  \quad \Leftrightarrow \quad s^{ea} = s^{ea}_{\mathrm{goal}} .
\end{equation}
Moreover, there exists a (possibly conservative) baseline decision
policy $\bar{\pi}$ such that, for any feasible initial pair
$(s^{ea}_0,x^{ea}_0)$,
\begin{itemize}
  \item the closed-loop system under $\bar{\pi}$ satisfies all
        state-input and safety constraints;
  \item the closed-loop state process reaches the goal mode
        $s^{ea}_{\mathrm{goal}}$ in finite time and then remains
        in this mode thereafter.
\end{itemize}
\end{assumption}

\begin{remark}
Given Assumption~\ref{baseline}, the terminal cost at the prediction
horizon can be defined as the cost-to-go under the baseline policy
until the goal mode is reached, i.e.,
\[
 J_H(s^{ea}_H, x^{ea}_H)
:= \sum_{t=0}^{\tau(x^{ea}_H)-1}
J^{ea}\!\Big(s^{ea}_{H+t},\,
\bar\pi\!\big(s^{ea}_{H+t},\, x^{ea}_{H+t},\, \mathcal{X}^{sa}_{H+t}\big)\Big).
\]
where $\tau(s^{ea}_H)$ denotes the first hitting time of
$s^{ea}_{\mathrm{goal}}$. This quantity is finite for all feasible
$(s^{ea}_H,x^{ea}_H)$.
\end{remark}

\begin{remark}
Assumption \ref{baseline} is standard and practically well-justified. It only requires the existence of a conservative baseline policy $\bar\pi$ that can drive the system to a goal state $s^{ea}_{\mathrm{goal}}$ in finite time while ensuring that all constraints are satisfied. In many autonomous systems, such a baseline policy can be constructed using simple rule-based strategies that guarantee constraint satisfaction by design. These rule-based strategies \cite{shu2025decision,xu2022integrated,vitelli2022safetynet} are commonly used in practice to ensure safe fallback behaviors in uncertain or dynamic environments and can thus serve as feasible baseline policies.  Therefore, although $\bar\pi$ is not required to be performance-optimal, the assumption that such a conservative baseline policy exists is both theoretically standard and very reasonable for the class of autonomous systems considered in this work.
\end{remark}

\begin{theorem}
Suppose Assumption \ref{baseline} holds. If at the initial time \( k = 0 \), the MPC optimization~(\ref{eq:etermpc}) can find an action sequence that satisfies all safety and feasibility requirements within the prediction horizon, then for all time steps \( k \geq 1 \), it can still find an action sequence that satisfies the same  requirements. In other words, the system maintains recursive feasibility.
\end{theorem}

\begin{proof}

Let the optimal action sequence and corresponding predicted state sequence at time \( k \) be:
\begin{align}
& a^*_k = \left\{ a^*_{k|k}, a^*_{k+1|k}, \dots, a^*_{k+H-1|k} \right\}, \\
& \boldsymbol{s}^*_k = \left\{ s^*_{k+1|k}, s^*_{k+2|k}, \dots, s^*_{k+H|k} \right\}.
\end{align}
At time \( k+1 \), we construct a feasible candidate 
\begin{align*}
\boldsymbol{\bar{a}}_{k+1} = \left\{ \bar{a}_{k+1|k+1}, \bar{a}_{k+2|k+1}, \dots, \bar{a}_{k+H|k+1} \right\},
\end{align*}
where
\begin{align}\label{eq:recaction}
  &  \bar{a}_{k+1+j|k+1} = a^*_{k+j+1|k}, \quad j = 0, \dots, H-2, \\
  &  \bar{a}_{k+H|k+1} = \bar{\pi}(s^*_{k+H|k}),
\end{align}
where \( \bar{\pi}(\cdot) \) is a known baseline policy from Assumption \ref{baseline} that satisfies all constraints.

Correspondingly, the candidate state trajectory at time $k+1$ is:
\begin{equation}
   \bar{s}_{k+1+j|k+1} = s^*_{k+j+1|k}, \quad j = 1, \dots, H-2,
\end{equation}
and the final state:
\begin{equation}\label{eq:recstate}
\bar{s}_{k+H|k+1} = f^{ea}(\bar{s}_{k+H-1|k+1}, \bar{a}_{k+H-1|k+1}).
\end{equation}

By Assumption~\ref{baseline}, the baseline action associated with the terminal reachable mode satisfies all constraints. Since all constraints are satisfied by the original optimal sequence and the terminal policy \(  \bar{\pi} \) ensures feasibility at the final step, the shifted sequence \((\bar{a}_{k+1}, \bar{\boldsymbol{s}}_{k+1})\) constitutes a feasible candidate at time \(k+1\). Hence, recursive feasibility is ensured.
\end{proof}

\begin{theorem}
Suppose Assumption~\ref{baseline} holds. Then the receding-horizon policy generated by the proposed MPC scheme stabilizes the EA HMDP with respect to the goal mode.
\end{theorem}

\begin{remark}
For clarity and representativeness, the stage cost $J^{ea}(s,a)$ is written as a compact form of the complete expression $J^{ea}(s^{ea}, a^{ea}, x^{ea}, \mathcal{X}^{sa})$, and the terminal cost $J_H(s)$ denotes $J_H(s^{ea}, x^{ea})$. 
This simplified form is mainly adopted in the stability analysis to facilitate theoretical derivations, while a similar state-action-dependent cost structure is also reflected in the cost design of the simulation studies.
\end{remark}

\begin{proof}
Define the optimal cost at time $k$ as:
\begin{align*}
V_{k}^* &=
\sum_{j=0}^{H-1}
J^{ea}(s^*_{k+j|k}, a^*_{k+j|k})
+ J_H(s^*_{k+H|k}) \\
&= J^{ea}(s^*_{k|k}, a^*_{k|k})
+ \sum_{j=1}^{H-1} J^{ea}(s^*_{k+j|k}, a^*_{k+j|k}) \\
&+ J_H(s^*_{k+H|k}).
\end{align*}
At time $k+1$, the candidate cost is:
\begin{align*}
V_{k+1} =&
\sum_{j=1}^{H-1} J^{ea}(s^*_{k+j|k}, a^*_{k+j|k})
+ J^{ea}(s^*_{k+H|k}, \pi(s^*_{k+H|k}))\\
&+ J_H(\bar{s}_{k+H+1|k+1}).
\end{align*}
Define the terminal difference as:
\begin{align*}
\Delta := & J^{ea}(s^*_{k+H|k}, \pi(s^*_{k+H|k}))
+ J_H(\bar{s}_{k+H+1|k+1})\\
&- J_H(s^*_{k+H|k}) \le 0.
\end{align*}

Thus:
\[
V_{k+1} - V_{k}^*
= \Delta - J^{ea}(s^*_{k|k}, a^*_{k|k}) \le 0.
\]
By optimality, $V^*_{k+1} \le V_{k+1}$, hence $V^*_{k+1} \le V^*_{k}$.
Since $J^{ea}(s,a) \ge 0$ for all non-goal states,
the system converges to the goal state $s^{ea}_{\mathrm{goal}}$ and asymptotic stability is ensured.
\end{proof}

\begin{remark}
Note that \(s^*_{k+H|k}\) and \(s^*_{k+H|k+1}\) both refer to the predicted state at time \(k+H\), though derived from different prediction steps. Therefore, the terminal cost \(\bar{J}(s^*_{k+H|k})\), designed as an upper bound of the future cost under the baseline policy, satisfies \(\Delta \leq 0\).
\end{remark}

\section{simulation}\label{IV}

To evaluate the effectiveness of the proposed decision-making algorithm,  we consider an autonomous driving scenario set in a three-lane road environment. The following sections provide the formulation of HMDPs for both the ego vehicle (EV) and SVs, design of the cost function, specification of safety constraints, and simulation studies and result analysis under various scenarios. 

At the low level, a rule-based MPC method is adopted to generate the EV’s acceleration and steering commands, following the implementation described in the Appendix of~\cite{wang2025high}. 
The vehicle dynamics are modeled using a nonlinear kinematic bicycle model, which captures realistic motion behavior during simulation.

\subsection{HMDP Formulation for SVs}

To capture the maneuver planning of the SVs, the $i$-th SV’s discrete action at time step $k$ is denoted by $a^{sv,i}_k = (a^{\text{lat},i}_k, a^{\text{long},i}_k)$, 
where the two elements represent the lateral and longitudinal action, respectively. 
The lateral component $a^{\text{lat},i}_k$ includes lane keeping, left lane change, and right lane change, while the longitudinal component $a^{\text{long},i}_k$ consists of constant-speed motion, accelerate, and decelerate. The Cartesian combination of these two maneuver dimensions yields a discrete action space containing nine representative maneuver actions, as summarized in Table~\ref{tab:actions}.

\begin{table}[ht]
\centering
\captionsetup{font=footnotesize, skip=6pt}
\caption{Maneuver actions of the SVs}
\label{tab:actions}
\setlength{\tabcolsep}{3pt}
\renewcommand{\arraystretch}{1.1}
\begin{tabular}{c c p{0.53\linewidth}}
\toprule
\textbf{Action} & 
\textbf{Symbol} $(a^{\text{lat},i}, a^{\text{long},i})$ & 
\textbf{Description} \\
\midrule
$a^{(1)}$ & (0, 0)   & Keep lane at constant speed \\
$a^{(2)}$ & (0, 1)   & Keep lane and accelerate \\
$a^{(3)}$ & (0, -1)  & Keep lane and decelerate \\
$a^{(4)}$ & (-1, 0)  & Change to left lane at constant speed \\
$a^{(5)}$ & (-1, 1)  & Change to left lane and accelerate \\
$a^{(6)}$ & (-1, -1) & Change to left lane and decelerate \\
$a^{(7)}$ & (1, 0)   & Change to right lane at constant speed \\
$a^{(8)}$ & (1, 1)   & Change to right lane and accelerate \\
$a^{(9)}$ & (1, -1)  & Change to right lane and decelerate \\
\bottomrule
\end{tabular}
\end{table}

Based on the lane topology and semantic map information, a discrete state set $s^{sv,i}_k = (s^{\text{lat},i}_k, s^{\text{long},i}_k)$ is constructed, where each state encodes the SV’s target lane and longitudinal motion maneuver. In the three-lane road environment,  lanes are indexed from left to right as Lane 1, Lane 2, and Lane 3. The complete set of discrete states is listed in Table~\ref{tab:states}. The transitions between different discrete states are illustrated in Table  \ref{tab:state_transition}.
\begin{table}[ht]
\centering
\captionsetup{font=footnotesize, skip=6pt}
\caption{Maneuver states of the SVs}
\setlength{\tabcolsep}{4pt} 
\label{tab:states}
\begin{threeparttable}
\begin{tabular}{ccc}
\toprule
\textbf{State} & \textbf{Symbol} $(s^{\text{lat},i}, s^{\text{long},i})$ & \textbf{Description} \\
\midrule
$s^{(1)}$ & (1, 0)   & Intending to be in Lane 1 with cruise speed \\
$s^{(2)}$ & (1, 1)   & Intending to be in Lane 1 with acceleration \\
$s^{(3)}$ & (1, -1)  & Intending to be in Lane 1 with deceleration \\
$s^{(4)}$ & (2, 0)   & Intending to be in Lane 2 with cruise speed \\
$s^{(5)}$ & (2, 1)   & Intending to be in Lane 2 with acceleration \\
$s^{(6)}$ & (2, -1)  & Intending to be in Lane 2 with deceleration \\
$s^{(7)}$ & (3, 0)   & Intending to be in Lane 3 with cruise speed \\
$s^{(8)}$ & (3, 1)   & Intending to be in Lane 3 with acceleration \\
$s^{(9)}$ & (3, -1)  & Intending to be in Lane 3 with deceleration \\
\bottomrule
\end{tabular}
\begin{tablenotes}
\footnotesize
\item Note: Each discrete state represents the EV’s high-level maneuver intent rather than its physical position. For example, a state such as “Intending to be in Lane 2 with acceleration” may indicate that the EV is already in Lane 2 or in the process of changing lanes towards it.
\end{tablenotes}
\end{threeparttable}
\end{table}
\begin{table}[ht]
\vspace{-3mm}
\centering
\captionsetup{font=footnotesize, skip=8pt}
\caption{State transition mapping of the SVs.}
\setlength{\tabcolsep}{5pt} 
\label{tab:state_transition}
\begin{threeparttable}
\begin{tabular}{cccccccccc}
\toprule
 & $a^{(1)}$ & $a^{(2)}$ & $a^{(3)}$ & $a^{(4)}$ & $a^{(5)}$ & $a^{(6)}$ & $a^{(7)}$ & $a^{(8)}$ & $a^{(9)}$ \\
\midrule
$s^{(1)}$ & $s^{(1)}$ & $s^{(2)}$ & $s^{(3)}$ & / & / & / & $s^{(4)}$ & $s^{(5)}$ & $s^{(6)}$ \\
$s^{(2)}$ & $s^{(2)}$ & / & $s^{(1)}$ & / & / & / & $s^{(5)}$ & / & $s^{(4)}$ \\
$s^{(3)}$ & $s^{(3)}$ & $s^{(1)}$  & / & / & / & / & $s^{(6)}$ & $s^{(4)}$ & / \\
$s^{(4)}$ & $s^{(4)}$& $s^{(5)}$ & $s^{(6)}$ & $s^{(1)}$ & $s^{(2)}$ & $s^{(3)}$ & $s^{(7)}$ & $s^{(8)}$ & $s^{(9)}$ \\
$s^{(5)}$ & $s^{(5)}$ & / & $s^{(4)}$ & $s^{(2)}$ & / & $s^{(1)}$ & $s^{(8)}$& / & $s^{(7)}$ \\
$s^{(6)}$ & $s^{(6)}$ & $s^{(4)}$ & / & $s^{(3)}$ & $s^{(1)}$ & / & $s^{(9)}$ & $s^{(7)}$ & / \\
$s^{(7)}$ & $s^{(7)}$ & $s^{(8)}$ & $s^{(9)}$ & $s^{(4)}$ & $s^{(5)}$ & $s^{(6)}$ & / & / & / \\
$s^{(8)}$ & $s^{(8)}$ & / & $s^{(7)}$ & $s^{(5)}$ & / & $s^{(4)}$ & / & / & / \\
$s^{(9)}$ & $s^{(9)}$ & $s^{(7)}$ &/ & $s^{(6)}$ & $s^{(4)}$ & / & / & / & / \\
\bottomrule
\end{tabular}
\begin{tablenotes}
\footnotesize
\item Note: Each row corresponds to the current discrete state $s^{(p)}$, and each column represents a possible action $a^{(q)}$. The table entry indicates the next state $s^{(r)}$ reached when action $a^{(q)}$ is executed at state $s^{(p)}$. For example, when the system is in state $s^{(1)}$, executing $a^{(1)}$ keeps it in $s^{(1)}$, while executing $a^{(2)}$ leads to $s^{(2)}$. ``/'' denotes infeasible transitions that violate physical or logical constraints.
\end{tablenotes}
\end{threeparttable}
\end{table}

The dynamic state vector of the $i$-th SV at time step $k$ is defined as
\[
\mathsf{\textbf{x}}_k^{sv,i} = \bigl[x_k^{sv,i},\; y_k^{sv,i},\; v_{x,k}^{sv,i} \bigr]^{\top},
\]
where $x_k^{sv,i}$ and $y_k^{sv,i}$ denote the longitudinal and lateral positions, respectively, 
and $v_{x,k}^{sv,i}$ represents the longitudinal velocity. The dynamic state evolves according to the following discrete-time prediction model:
\[
\mathsf{\textbf{x}}_{k+1|k}^{sv,i}
= \mathsf{A}_k^{sv,i}(a_k^{\text{lat},i})\, \mathsf{\textbf{x}}_{k|k}^{sv,i}
+ \mathsf{B}_k^{sv,i}(a_k^{\text{lat},i})\, \mathsf{F}_k^i + w_k^i,
\]
where $\mathsf{A}_k^{sv,i}$ and $\mathsf{B}_k^{sv,i}$ are the state and input matrices determined by the discrete lateral action $a_k^{\text{lat},i}$. 
The disturbance $w_k^i$ is modeled as zero-mean Gaussian noise, and its propagation over the prediction horizon is computed using~\eqref{eq:coinvpro}.

The vector
\[
\mathsf{F}_k^i =
\begin{bmatrix}
a_x(s_{k}^{\text{long},i}) \\[2pt]
u_y(a_k^{\text{lat},i})
\end{bmatrix}
\]
represents a hypothetical control signal, where $a_x$ denotes the nominal longitudinal acceleration determined by the current longitudinal mode $s_k^{\text{long},i}\!\in\!\{-1,0,1\}$. Each mode corresponds to a representative average acceleration value that reflects the driver's decision to decelerate, maintain speed, or accelerate, while $u_y$ is a virtual lateral control generated by a proportional-derivative (PD) law:
\[
u_{y,k}^i = K_1\!\bigl(y_{\mathrm{ref}}-y_k^{sv,i}\bigr)
          + K_2\!\bigl(v_{y,\mathrm{ref}}-v_{y,k}^{\,i}\bigr),
\]
with $ v_{y,k}^{\,i}$ is the lateral velocity. 

The action-dependent matrices are defined as
\[
\mathsf{A}^{sv}(\text{LK})=
\begin{bmatrix}
1 & 0 & T_h\\
0 & 1 & 0\\
0 & 0 & 1
\end{bmatrix},\quad
\mathsf{B}^{sv}(\text{LK})=
\begin{bmatrix}
\tfrac12 T_h^2 & 0\\
0              & 0\\
T_h            & 0
\end{bmatrix},
\]
and for lane-change maneuvers (both left and right):
\begin{align*}
&\mathsf{A}^{sv}(\text{LC})=
\begin{bmatrix}
1 & 0 & T_h\\
0 & 1-\tfrac12 K_1 T_h^2 & 0\\
0 & 0 & 1
\end{bmatrix},\quad \\
&\mathsf{B}^{sv}(\text{LC})=
\begin{bmatrix}
\tfrac12 T_h^2 & 0\\
0              & \tfrac12 T_h^2\\
T_h            & 0
\end{bmatrix}.  
\end{align*}

Accordingly, the pair $\bigl(A_k^{sv,i}, B_k^{sv,i}\bigr)$ is selected as
\[
\bigl(\mathsf{A}_k^{sv,i}, \mathsf{B}_k^{sv,i}\bigr)=
\begin{cases}
\bigl(\mathsf{A}_{\text{LK}}, \mathsf{B}_{\text{LK}}\bigr), & a^{\text{lat},i}_k=0,\\[6pt]
\bigl(\mathsf{A}_{\text{LC}}, \mathsf{B}_{\text{LC}}\bigr), & a^{\text{lat},i}_k \in \{-1, 1\}.
\end{cases}
\]

Based on the above model, the potential maneuvers and corresponding dynamic evolution of the SVs are predicted and subsequently utilized for safety assessment.

\begin{remark}
In this formulation, the lateral action $a_k^{\text{lat},i}$ and the longitudinal state $s_k^{\text{long},i}$ jointly serve as the basis for predicting the state evolution, rather than relying solely on either actions or states. This design choice reflects the underlying semantics of driving behaviors: lateral dynamic states evolution depend on whether a lane-change action occurs, rather than on the specific lane identity; whereas longitudinal dynamic states evolution depends on the ongoing maneuver state, being in an accelerating or decelerating phase, rather than on the instant of initiating such actions. A similar principle is reflected in the subsequent cost-function design.
\end{remark}

\subsection{HMDP Formulation for EV}

The MDP formulation of the EV follows a similar modeling principle to that of the SVs, encompassing the definition of discrete maneuver states, admissible actions, and transition rules. 
Specifically, the discrete state and action of the EV are denoted as
\[
s_k^{ea} = (s_k^{\mathrm{lat}},\, s_k^{\mathrm{long}}), \quad 
a_k^{ea} = (a_k^{\mathrm{lat}},\, a_k^{\mathrm{long}}).
\]
Note that the subscript \(i\) is omitted for the EV since there is only one EV. The subscript \(i\) is reserved for SVs to distinguish multiple surrounding agents. Thus, the presence or absence of \(i\) indicates whether the variable refers to SVs or the EV.

The EV’s discrete states and actions share the same structural definitions as those of the SVs, following Tables~\ref{tab:actions} and~\ref{tab:states}. Their transitions are governed by the same rule set summarized in Table~\ref{tab:state_transition}.

The dynamic state of the EV is subsequently expressed as
\[
\mathsf{\textbf{x}}_k^{ev} = \bigl[x_k^{ev},\; y_k^{ev},\; v_{x,k}^{ev} \bigr]^{\top},
\]
where $x_k^{ev,i}$ and $y_k^{ev,i}$ denote the longitudinal and lateral positions of EV, respectively, 
and $v_{x,k}^{ev,i}$ is its longitudinal velocity. 

Given a high-level sampling period \(T_h\), the state evolution for prediction is described by
\begin{equation}\label{eq:evdynamic}
\mathsf{\textbf{x}}_{k+1|k}^{ev}
=
\mathsf{A}^{ev}\mathsf{\textbf{x}}_{k|k}^{ev}
+
\mathsf{B}^{ev}\!\bigl(a_k^{\mathrm{lat}},\,s_{k+1}^{\mathrm{long}}\bigr),   
\end{equation}
where the system and input matrices are defined as
\[
\mathsf{A}^{ev} =
\begin{bmatrix}
1 & 0 & T_h \\
0 & 1 & 0  \\
0 & 0 & 1
\end{bmatrix},
\quad
\mathsf{B}^{ev}\!\bigl(a_k^{\mathrm{lat}},s_{k+1}^{\mathrm{long}}\bigr) =
\begin{bmatrix}
\dfrac{1}{2}\,s_{k+1}^{\mathrm{long}}\,a_{\mathrm{ave}}\,T_h^{2} \\[3pt]
\Delta y_k \\[3pt]
s_{k+1}^{\mathrm{long}}\,a_{\mathrm{ave}}\,T_h
\end{bmatrix},
\]
with \(a_{\mathrm{ave}}\) denoting the nominal longitudinal acceleration magnitude.
The lateral displacement \(\Delta y_k\) is generated following a fifth-order polynomial lane-change profile:
\[
\Delta y_k =
    a_k^{\mathrm{lat}}\,
    r_d\bigl[c_1\tau_3(k)+c_2\tau_4(k)+c_3\tau_5(k)\bigr],
 \]
where \(r_d\) is the lane width, \(c_1,c_2,c_3\) are fixed coefficients, and \(\tau_i(k)\) denotes the normalized polynomial basis function of order \(i\) corresponding to the longitudinal maneuver.

In summary, the above equations couple the discrete variables
\(\bigl(a_k^{\mathrm{lat}},\,s_{k+1}^{\mathrm{long}}\bigr)\) with the dynamics,
thereby enabling consistent motion prediction across all high-level maneuver modes.

\subsection{Cost Function}

The cost function serves as an interface between human driving preferences and the EV's decision-making process. It encodes subjective priorities across different maneuver combinations and transfers them to the autonomous vehicles. As discussed earlier, the EV's state evolution is primarily governed by the lateral action \(a_k^{\mathrm{lat}}\) and the longitudinal state \(s_k^{\mathrm{long}}\).  To systematically represent all possible combinations of these variables and ensure a unified, unbiased evaluation during optimization, a one-hot representation is adopted for both the lateral and longitudinal variables.

For the longitudinal dimension, the one-hot vector 
\(\bigl(s_{k+1}^{\mathrm{acc}},\,s_{k+1}^{\mathrm{cru}},\,s_{k+1}^{\mathrm{dec}}\bigr)\) 
corresponding to acceleration, cruising, and deceleration, respectively, is defined,  
We then impose the following mode-consistency constraint:
\[
s_{k+1}^{\mathrm{long}} = s_{k+1}^{\mathrm{acc}} - s_{k+1}^{\mathrm{dec}},
\]
which ensures that the longitudinal mode used in the prediction dynamics \eqref{eq:evdynamic} matches the discrete mode selected by the optimizer, 
where selecting \(\bigl(s_{k+1}^{\mathrm{acc}},\,s_{k+1}^{\mathrm{cru}},\,s_{k+1}^{\mathrm{dec}}\bigr)\) activates the corresponding dynamic submodel.

Similarly, for the lateral dimension, a one-hot vector 
\(\bigl(a_k^{\mathrm{left}},\, a_k^{\mathrm{keep}},\, a_k^{\mathrm{right}}\bigr)\),
corresponding to left lane change, lane keeping, and right lane change, respectively, is defined. 
The following constraint is then imposed:
\[
a_k^{\mathrm{lat}} = a_k^{\mathrm{right}} - a_k^{\mathrm{left}},
\]
which ensures that the lateral action used in the prediction dynamics \eqref{eq:evdynamic} matches the action variable selected in the cost function, where choosing \(\bigl(a_k^{\mathrm{left}}, a_k^{\mathrm{keep}}, a_k^{\mathrm{right}}\bigr)\) activates the corresponding dynamic submodel.

\begin{remark}
In the MPC optimization, directly using the discrete variables 
\(a_k^{\mathrm{lat}}, s_{k+1}^{\mathrm{long}} \in \{-1,0,1\}\) may introduce numerical bias, 
as the cost minimization tend to  favor negative-valued modes (e.g., \(a_k^{\mathrm{lat}}=-1\)), which appear to lower the objective value.
The one-hot formulation removes this imbalance by assigning binary indicators (1 for the selected mode and 0 otherwise), enabling all maneuver options to be evaluated fairly within a unified cost structure.
\end{remark}

The cost function is then defined over the prediction horizon as
\begin{equation}\label{eq:cost_coupled}
   J^{ev} = \sum_{k=1}^{H} \sum_{m=1}^{3} \sum_{n=1}^{3} c_{m,n} \, a_k^{(n)} \, s_{k+1}^{(m)},
\end{equation}
where:
\begin{itemize}
    \item \( a_k^{(n)} \in \{ a_k^{\mathrm{left}},\ a_k^{\mathrm{keep}},\ a_k^{\mathrm{right}} \} \) denotes the lateral action at time step \(k\);
    \item \( s_{k+1}^{(m)} \in \{ s_{k+1}^{\mathrm{acc}},\ s_{k+1}^{\mathrm{cru}},\ s_{k+1}^{\mathrm{dec}} \} \) denotes the longitudinal maneuver state at the next time step;
    \item \( c_{m,n} \) is the cost coefficient associated with the \((m,n)\)-th maneuver combination.
\end{itemize}

These cost coefficients \( c_{m,n} \) encode the relative preference among different maneuver pairs in the decision-making process. Stable and smooth driving behaviors, such as lane keeping combined with cruising, are assigned smaller coefficients, making them more likely to be selected. Conversely, aggressive combinations, such as lane changing while decelerating, are assigned larger coefficients, discouraging their selection unless necessary.

\subsection{Safety Constraints}

Safety constraints are formulated to guarantee that the EV maintains a sufficient longitudinal separation from SVs. These constraints are conditionally activated based on the maneuver states. Specifically, a safety constraint is enforced whenever the EV and a SV share the same intended target lane, i.e.,
\[
s^{\mathrm{lat},i}_{k+1|k} = s^{\mathrm{lat}}_{k+1|k}.
\]

We begin with the standard chance-constraint formulation defined in~\eqref{eq:chance_constraints}, and specify the constraint function as
\[
h\bigl(x_{k+1|k}^{ev}, x_{k+1|k}^{sv,i}\bigr)
= x_{k+1|k}^{ev} - x_{k+1|k}^{sv,i} - d_{\text{safe}}.
\]
Accordingly, the chance constraint is expressed as
\[
\mathbb{P}\!\left(
    h\bigl(x_{k+1|k}^{ev}, x_{k+1|k}^{sv,i}\bigr) \ge 0
\right) \ge 1 - \epsilon,
\]
which enforces that the predicted longitudinal separation satisfies the safety margin \( d_{\text{safe}} \) with confidence level \( 1-\epsilon \).

As \(x_{k+1|k}^{sv,i}\) follows a Gaussian distribution while \(x_{k+1|k}^{ev}\) is deterministic, the relative distance
\(\Delta x_{k+1|k} = x_{k+1|k}^{ev} - x_{k+1|k}^{sv,i}\)
also follows a Gaussian distribution with mean and variance
\[
\mu_{\Delta,k+1|k} = x_{k+1|k}^{ev} - \bar{x}_{k+1|k}^{sv,i},
\qquad
Q_{\Delta,k+1|k} = Q_{k+1|k}^{sv,i}.
\]
where \( \bar{x}_{k+1|k}^{sv,i} \) represents the predicted mean position of the \(i\)-th SV. Consequently, the chance constraint can be reformulated into the following deterministic inequality:
\[
\mu_{\Delta,k+1|k}
\ge d_{\text{safe}} + z_\epsilon \sqrt{ Q_{\Delta,k+1|k} }.
\]
When the EV is predicted to be following the SV, the constraint takes the symmetric form:
\[
-\mu_{\Delta,k+1|k}
\ge d_{\text{safe}} + z_\epsilon \sqrt{ Q_{\Delta,k+1|k} }.
\]
This completes the deterministic reformulation of the safety constraint, enabling its direct incorporation into the MPC optimization framework.

\subsection{Case 1: Validation under Lateral Maneuver Uncertainty}
In this scenario, the SVs are assumed to maintain a constant longitudinal velocity to isolate the effect of lateral uncertainty, and the resulting EV decisions are analyzed accordingly.

\subsubsection{Scenario Description}

 A three-lane urban expressway scenario is examined, where the EV (depicted in red) initially cruises in the middle lane at a constant speed. Two SVs are initialized in the left and right lanes, respectively. Each SV has two potential lateral maneuvers: (i) to maintain cruising in its current lane, or (ii) to perform a lane change into the middle lane. Their predicted trajectories under different maneuvers are illustrated in different colors. The ellipses along the predicted trajectories visualize the propagated uncertainty, and are constructed from the covariance matrices associated with each predicted state.

\subsubsection{Key Parameters}
The key simulation parameters and initial conditions of the vehicles are summarized in Tables~\ref{tab:case1_parameters} and~\ref{tab:initial_conditions}, respectively.

\begin{table}[ht]
\centering
\captionsetup{font=footnotesize, skip=6pt}
\caption{Simulation and vehicle parameters (Case~1)}
\setlength{\tabcolsep}{4pt}
\label{tab:case1_parameters}
\begin{threeparttable}
\begin{tabular}{lll}
\toprule
\textbf{Parameter} & \textbf{Description} & \textbf{Value} \\ 
\midrule
$T_{\mathrm{sim}}$  & Simulation duration                        & 50~s \\
$T_l$               & Low-level control sampling period          & 0.2~s \\
$T_h$               & High-level decision update period          & 0.8~s \\
$H$                 & High-level prediction horizon              & 3~steps (2.4~s) \\
$w_{\mathrm{lane}}$ & Lane width                                 & 4~m \\
$d_{\mathrm{safe}}$ & Minimum longitudinal safe distance         & 40~m \\
$K_1,\,K_2$         & PD gains for lateral control               & 3,\;1 \\
$a_{\mathrm{avg}}$  & Nominal acceleration magnitude\tnote{$\ast$} & 2~m/s\textsuperscript{2} \\
$\delta$            & Cumulative probability threshold           & $10^{-5}$ \\ 
$c_{\mathrm{mn}}$   & Cost coefficients in~\eqref{eq:cost_coupled}
& $\begin{bmatrix}
6 & 5 & 6 \\
2 & 0 & 2 \\
11 & 9 & 11
\end{bmatrix}$ \\ 
\bottomrule
\end{tabular}
\begin{tablenotes}
\footnotesize
\item[$\ast$] Positive value denotes acceleration; a negative value (not shown) would indicate the same magnitude of deceleration.
\end{tablenotes}
\end{threeparttable}
\end{table}
\begin{table}[ht]
\vspace{-3mm} 
\centering
\captionsetup{font=footnotesize, skip=6pt}
\caption{Initial conditions (Case~1)}
\setlength{\tabcolsep}{4pt}
\label{tab:initial_conditions}
\begin{threeparttable}
\begin{tabular}{lllll}
\toprule
\textbf{Vehicle} & \textbf{Initial Lane} & $x_0$ (m) & $y_0$ (m) & $v_0$ (m/s) \\ \midrule
EV (red)       & Lane 2 & 35  & 0   & 20 \\
SV1 (blue)   & Lane 1   & 100 &  4  & 15 \\
SV2 (green)  & Lane 3  & 60  & -4  & 19 \\ \bottomrule
\end{tabular}
\end{threeparttable}
\end{table}

\subsubsection{Observations and Results}

The simulation results demonstrate the EV’s adaptive maneuvering behavior in response to the SVs’ lateral-maneuver uncertainty. 

Figure~\ref{fig:case1_snapshots} presents six snapshots of the traffic scene. Each snapshot corresponds to the state resulting from the decisions taken at the sampling instants marked by yellow circles in Fig.~\ref{fig:case1_actiontime}. Initially, the EV (red) cruises in Lane~2 at a constant speed. Because both SVs are predicted to merge into Lane~2, the EV refrains from a straightforward overtaking maneuver. 
As the simulation progresses, the right-lane vehicle (green) gradually passes the left-lane vehicle (blue). Aware of this evolution, the EV performs a lane change to Lane~3, positions itself behind the faster SV, and temporarily follows it to benefit from its higher speed.  
After the green SV overtakes the blue SV, the EV returns to Lane 1 and resumes cruising.  

\begin{figure}[htbp]
    \captionsetup{font=footnotesize}
    \centering
    \includegraphics[width=0.49\textwidth]{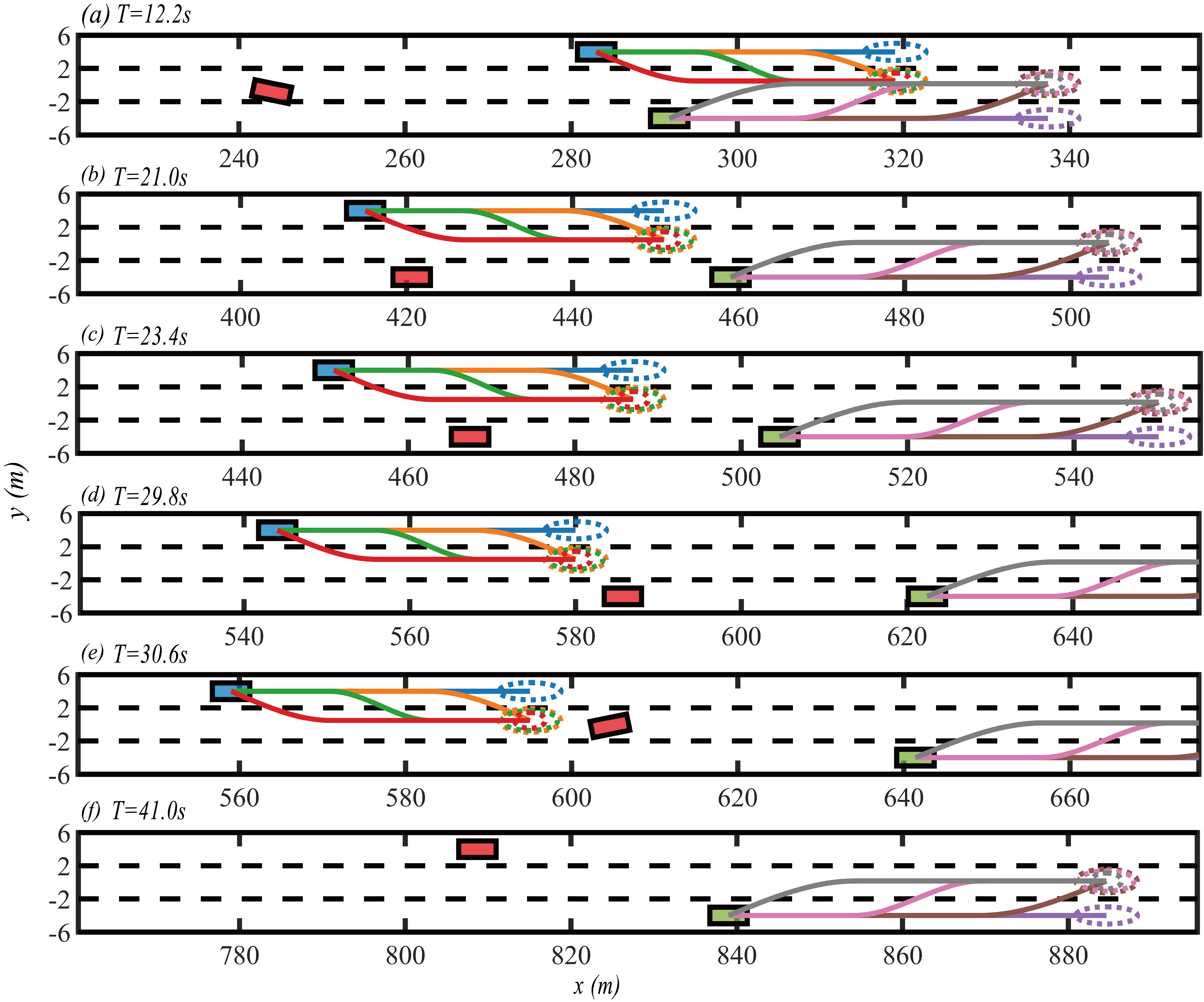}
    \caption{Snapshots of the EV (red) at six representative time instants. 
Each snapshot shows the vehicle state at $t_{k+1}$ resulting from the decision 
computed at the sampling instant $t_k$ marked in Fig.~\ref{fig:case1_actiontime}. 
Subfigures (a)–(f) correspond to 12.2, 21.0, 23.4, 29.8, 30.6, and 41.0~s, respectively.
    }
    \label{fig:case1_snapshots}
\end{figure}

Fig.~\ref{fig:case1_actiontime} depicts the evolution of the EV’s discrete MDP states and selected actions in Case~1. The upper subplot shows the lateral decisions, while the lower subplot displays the longitudinal decisions. The yellow circles mark the sampling instants at which decisions are computed. The resulting continuous states at the subsequent sampling instants are visualized in Fig.~\ref{fig:case1_snapshots}. The complete trajectories of the EV (red) and both SVs (blue and green) are given in Fig.~\ref{fig:case1_trajectory}. 

\begin{figure}[htbp]  \captionsetup{font=footnotesize}   \centering   \includegraphics[width=0.49\textwidth]{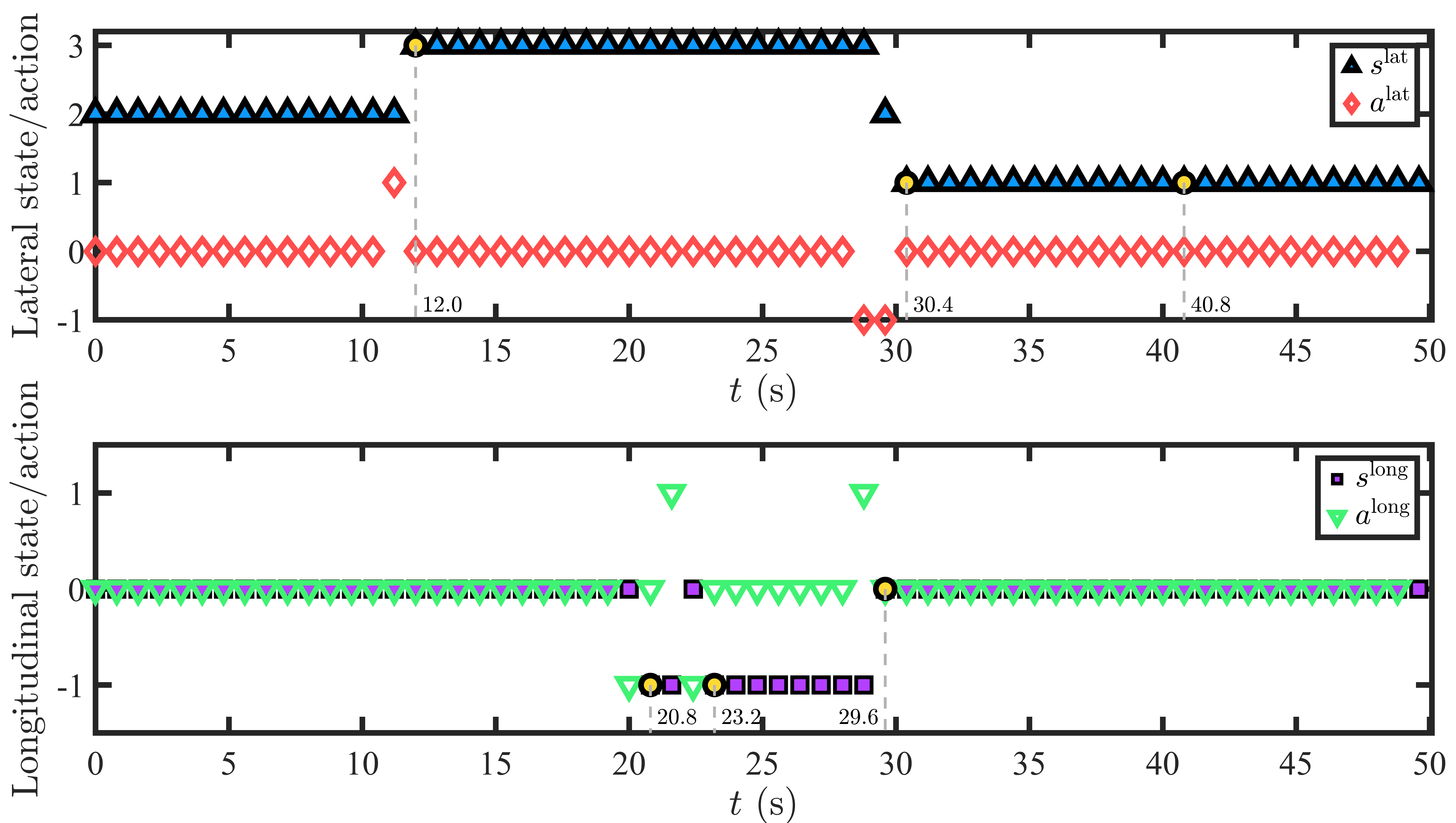}
    \caption{Time histories of the EV’s lateral (top) and longitudinal (bottom) MDP states and actions in Case~1. Yellow circles indicate the sampling instants at which  the six decisions are taken, while Fig.~\ref{fig:case1_snapshots} shows the resulting 
states at the next sampling instants.}

    \label{fig:case1_actiontime}
\end{figure}

\begin{figure}[htbp]
\captionsetup{font=footnotesize}
    \centering    \includegraphics[width=0.49\textwidth]{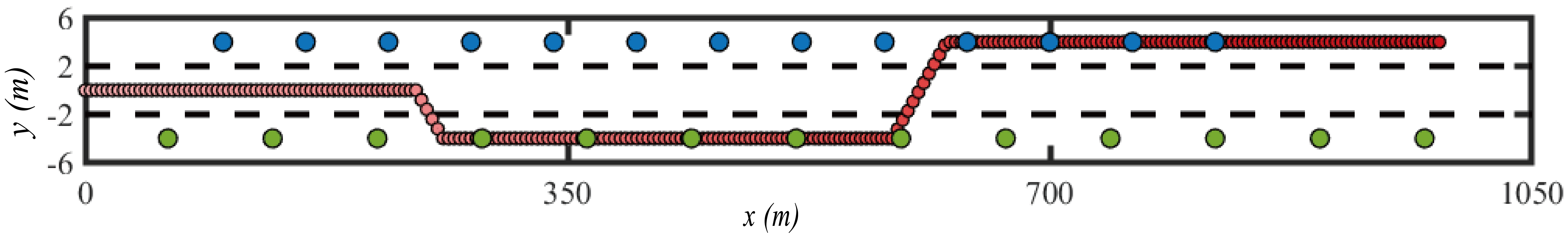}
    \caption{Complete EV trajectory in Case~1. }
    \label{fig:case1_trajectory}
\end{figure}

In summary, Case 1 focuses on the framework’s response to lateral-maneuver uncertainty from SVs.
The results show that the EV can anticipate and adapt to uncertain lane-change behaviors of neighboring vehicles while maintaining safe distances and smooth motion, confirming the framework’s ability to handle lateral interaction uncertainty effectively.

\subsection{Case 2: Validation under Longitudinal Maneuver Uncertainty}

In Case 1, the focus was on the uncertainty in the SVs’ lateral maneuvers. In this case, the analysis shifts to longitudinal uncertainty, specifically variations in velocity, to examine how the EV adapts its decisions to changing traffic dynamics.

\subsubsection{Scenario Description}

This scenario is set on a three-lane expressway, where Lane~3 is typically reserved for heavy trucks. As a result, the EV tends to avoid this lane unless under unavoidable circumstances. Two SVs are involved: one drives ahead of the EV in Lane~2, while the other occupies the adjacent Lane~1. Each vehicle is assumed to exhibit one of three possible longitudinal maneuvers: cruising, accelerating, or decelerating. The EV must evaluate the potential maneuvers of both SVs and determine whether to follow, overtake, or maintain its current maneuver.

\subsubsection{Key Parameters}
Most parameters remain consistent with Case~1; however, the minimum safety distance is increased to 75~m, while the low-level sampling period and the high-level decision update period are adjusted to 0.15~s and 0.6~s, respectively, to account for the higher vehicle speeds in this scenario. The initial conditions of the vehicles are summarized in Table~\ref{tab:initial_conditions_case2}.

\begin{table}[ht]
\centering
\captionsetup{font=footnotesize, skip=6pt}
\caption{Initial conditions (Case~2)}
\setlength{\tabcolsep}{4pt}
\label{tab:initial_conditions_case2}
\begin{threeparttable}
\begin{tabular}{lllll}
\toprule
\textbf{Vehicle} & \textbf{Initial Lane} & $x_0$ (m) & $y_0$ (m) & $v_0$ (m/s) \\ \midrule
EV (red)& Lane 2 &  35  &   0  & 30 \\
SV1 (blue)   & Lane 1   &   0  &   4  & 25 \\
SV2 (green)  & Lane 2  & 150  &  0  & 25 \\ \bottomrule
\end{tabular}
\end{threeparttable}
\end{table}

\subsubsection{Observations and Results}

Fig.~\ref{fig:longitudinal_strip_seq} and Fig.~\ref{fig:velocity_profile} jointly illustrate the EV’s behavior under longitudinal uncertainty arising from the SVs’ velocity variations. The EV adapts its strategy dynamically in response to these changes.  

In Fig.~\ref{fig:longitudinal_strip_seq}, the colored strips represent the predicted longitudinal trajectories of SVs under different maneuver modes: green for deceleration, purple for cruising, and red for acceleration. The transparency of each strip encodes the likelihood of the corresponding mode, with deeper color indicating higher probability. 

We assume that the SV1 (blue) initially cruises with a high probability of 0.9 in Lane 1, while the acceleration and deceleration modes each have a probability of 0.05 before $t = 7.8~ \text{s}$. After $t = 7.8 ~\text{s}$, the dominant mode switches to acceleration with 0.9 probability. For the SV2 (green) in Lane 2, the cruise mode remains dominant throughout the entire episode with a fixed probability of 0.8, while the acceleration and deceleration modes each retain a probability of 0.1.

As shown in Fig.~\ref{fig:longitudinal_strip_seq},
initially, the EV accelerates to initiate an overtaking maneuver while the rear vehicle in the target lane (SV1) maintains a cruising speed.
During the lane-change process, SV1 continues to accelerate and close in from behind, prompting the EV to perform a secondary acceleration to maintain a safe separation before completing the maneuver.
After overtaking the slower leading vehicle (SV2), the EV gradually reduces its speed to return toward the desired cruising velocity.
When the faster SV1 approaches from behind, the EV proactively returns to the middle lane, effectively yielding the overtaking lane while maintaining sufficient safety margins.

The corresponding velocity profile is shown in Fig.~\ref{fig:velocity_profile} (solid lines), which captures these adaptive behaviors, with five key decision points, marked by yellow circles, indicating distinct speed transitions associated with maneuver changes. The complete spatial trajectory is depicted in Fig. ~\ref{fig:case2_trajectory}, where the path marked with red circles confirms that the EV successfully executes the overtaking maneuver and returns to the original lane.

In summary, Case 2 demonstrates that the proposed framework is also capable of producing optimal decisions under longitudinal uncertainty, ensuring both safety and efficiency.

The EV’s decision-making mechanism has been validated under both isolated lateral and longitudinal uncertainties in Case 1 and Case 2, respectively. Since the longitudinal and lateral components are formulated within the same decision-making architecture, considering them jointly only increases the dimensionality of the optimization. Consequently, the proposed framework can be readily extended to scenarios in which lateral and longitudinal uncertainties coexist.

\begin{figure}[htbp]
   \captionsetup{font=footnotesize}
    \centering
    \begin{subfigure}[b]{0.5\textwidth}
      \includegraphics[width=\textwidth]{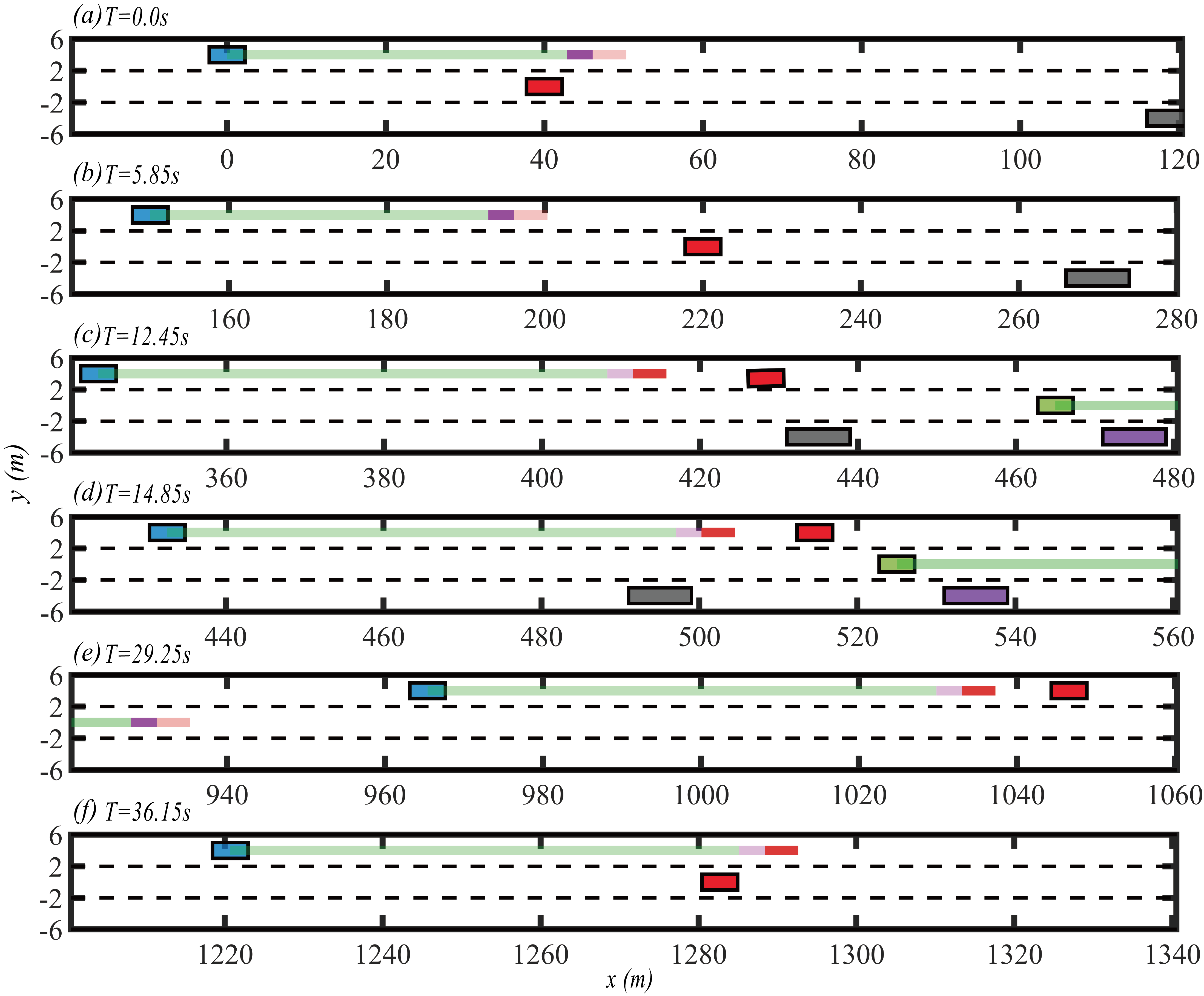}
    \end{subfigure}
    \caption{Snapshots of EV maneuver at six key decision moments under longitudinal uncertainty. Subfigures (a)-(f) correspond to the time points 0.00~s, 5.85~s, 12.45~s, 14.85~s, 29.25~s, and 36.15~s, respectively.}
    \label{fig:longitudinal_strip_seq}
\end{figure}
\begin{figure}[htbp]
   \captionsetup{font=footnotesize}
    \centering
    \begin{subfigure}[b]{0.49\textwidth}
      \includegraphics[width=\textwidth]{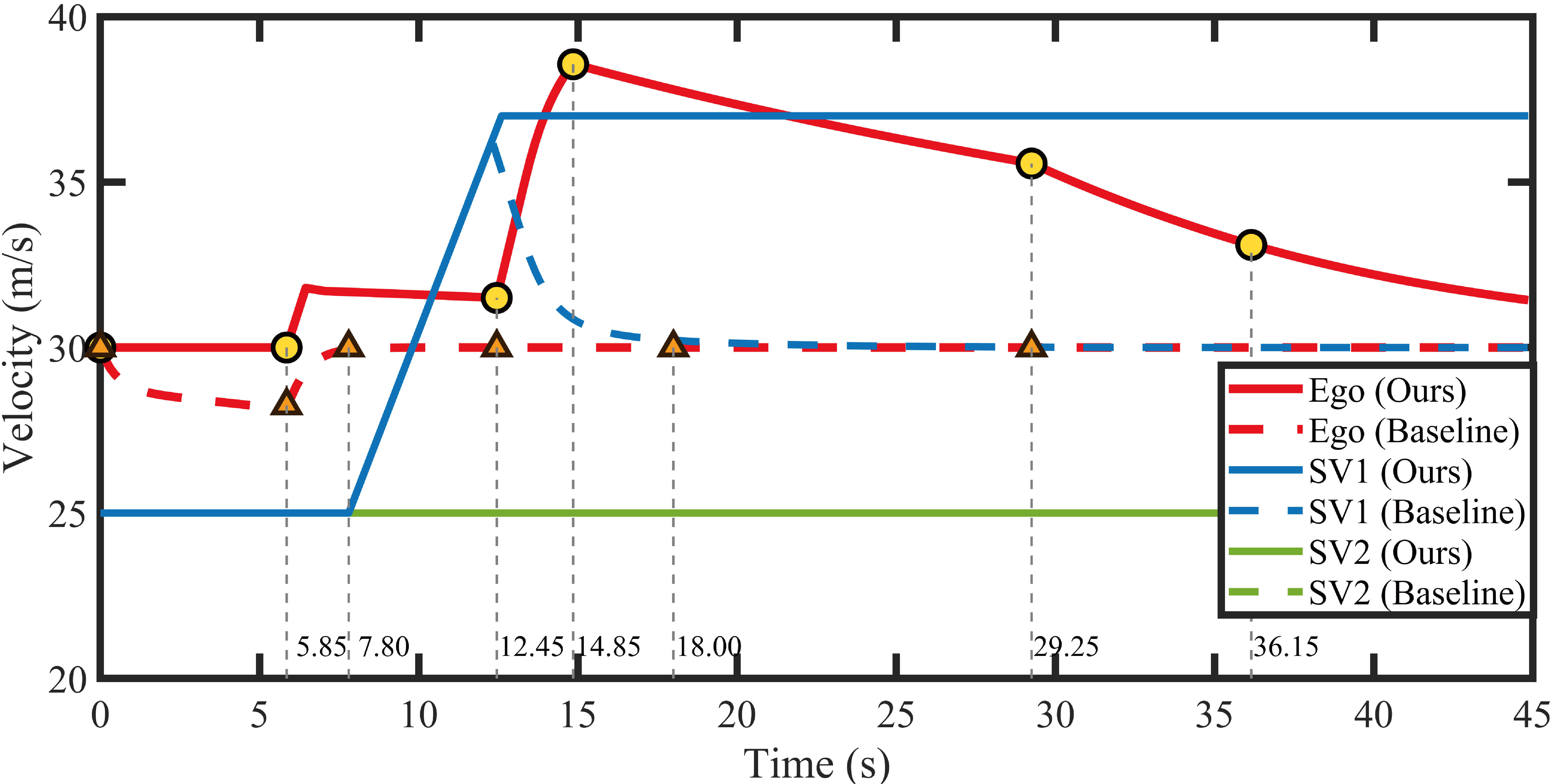}
    \end{subfigure}
    \caption{Comparison of velocity profiles between the proposed HMDP-MPC framework (solid lines) and the rule-based baseline (dashed lines). The markers indicate key decision moments corresponding to the snapshots in Fig. ~\ref{fig:longitudinal_strip_seq} and Fig.~\ref{fig:rulebased_snapshots}.}
    \label{fig:velocity_profile}
\end{figure}

\begin{figure}[htbp]
   \captionsetup{font=footnotesize}
    \centering
    \includegraphics[width=0.49\textwidth]{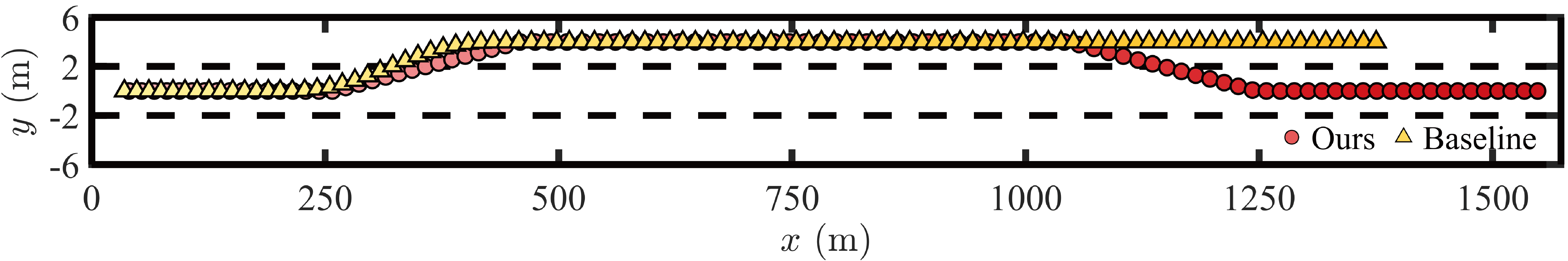}
    \caption{Comparison of EV trajectories in Case 2. The red circles denote the EV trajectory under the proposed HMDP-MPC framework, and the yellow triangles denote the EV trajectory under the rule-based baseline.}
    \label{fig:case2_trajectory}
\end{figure}

\subsubsection{Comparison with Rule-Based Method}

To highlight the performance distinctions between the proposed framework and conventional reactive strategies, we compare it against a rule-based method based on the IDM \cite{treiber2000congested} and MOBIL models \cite{kesting2007general} \footnote{The implementation details and simulation code for the rule-based baseline 
are available at: \url{https://github.com/SIYUANLI2023/IDM-MOBIL/tree/main}.}. Fig.~\ref{fig:velocity_profile}-Fig. \ref{fig:rulebased_snapshots} present the comparative results under this baseline.

Fig.~\ref{fig:rulebased_snapshots} illustrates six snapshots of the traffic scene at key decision moments under the rule-based baseline. 
Initially, the EV maintains cruising in Lane~2 while SV1 in the adjacent lane is cruising. 
At $t\!=\!5.85$~s,  the EV initiates a lane-change maneuver toward Lane~1. 
When $t\!=\!7.8$~s, SV1 has begun to accelerate (see Fig.~\ref{fig:velocity_profile} dashed lines) and is continuously closing in from behind while the EV is lane changing. Despite the decreasing gap, the EV continues the lane-change maneuver without adjustment, leading SV1 to perform an abrupt deceleration to avoid a potential collision. 
The EV then completes the lane change and overtakes the slower leading vehicle (SV2). 
However, after overtaking, it continues cruising in the overtaking lane without autonomously returning to its original lane (see Fig.~\ref{fig:case2_trajectory} yellow triangles).

Compared with the proposed HMDP-MPC framework, several distinct differences can be identified:
\paragraph{Efficiency}
During the 45 s simulation, the EV and SV1 under the rule-based baseline cover approximately 1250~m and 1150~m, respectively, whereas in the proposed framework they reach about 1400~m and 1300~m, respectively. 
This indicates that the proposed approach achieves higher motion efficiency by anticipating future interactions and maintaining smoother velocity evolution without oscillatory corrections.

\paragraph{Underlying mechanism difference}

Rule-based strategies rely on explicitly defined instructions that specify what an agent should do under given conditions (e.g., “if the vehicle ahead is close, then decelerate”). 
In the considered rule-based baseline, no rule is defined for situations where a vehicle rapidly approaches from behind during lane change process. 
As a result, the EV fails to adjust its behavior in time, forcing the following vehicle to perform abrupt deceleration and leaving the EV occupying the overtaking lane without yielding. 

In contrast, the proposed framework specifies what the agent must not do by formulating safety constraints (e.g., “the inter-vehicle distance must not fall below a threshold”). 
These constraints are inherently embedded in the optimization process, allowing the EV to automatically select feasible actions that satisfy safety requirements without additional decision rules. 
Even when the following vehicle accelerates unexpectedly, the EV naturally adapts its maneuver within the safe constraint set, thereby demonstrating a higher level of autonomy.

Furthermore, the proposed framework enables proactive responses to potential interactions before they occur. This predictive capability prevents sudden braking or unnecessary lane occupation, resulting in both safer and more efficient maneuvers compared with the reactive rule-based controller.

\begin{figure}[t]
   \captionsetup{font=footnotesize}
    \centering    \includegraphics[width=0.49\textwidth]{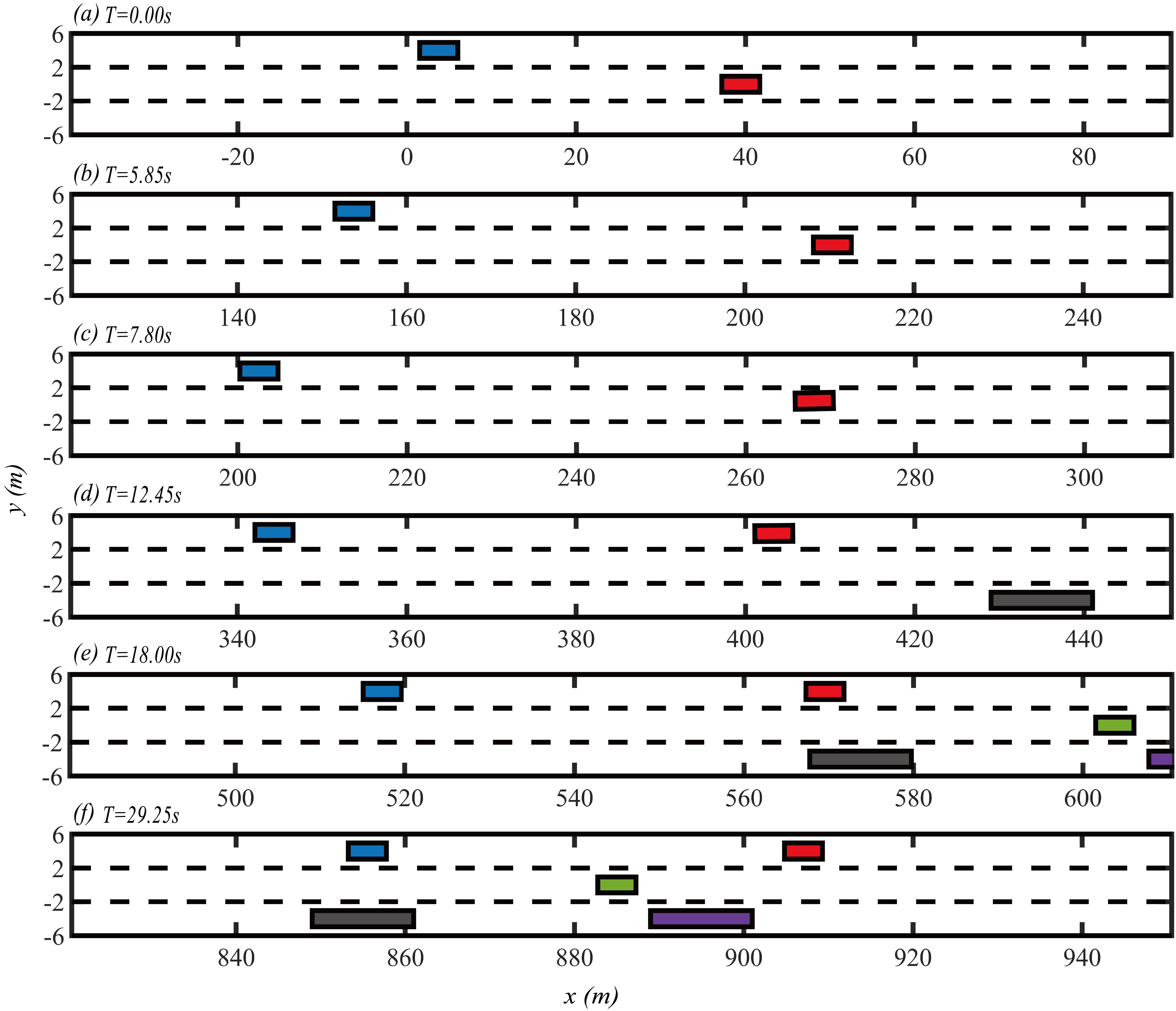}
    \caption{Snapshots of EV maneuver under the rule-based method at six decision moments. Subfigures (a)-(f) correspond to the time points 0.00~s, 5.85~s, 7.80~s, 12.45~s, 18.00~s, and 29.25~s, respectively.}
    \label{fig:rulebased_snapshots}
\end{figure}

\subsection{Case 3: Risk Threshold Analysis in Urban Environments}

To systematically examine the effect of the risk threshold $\epsilon$ on the EV's decision-making process, we construct an urban scenario and compare the resulting maneuvers under different risk levels.

\subsubsection{Scenario Description}

The scenario emulates an urban road segment characterized by lower velocities. The EV is initially positioned in Lane 1, accompanied by two SVs. Owing to the complexity of urban environments, shorter control and decision intervals are adopted to enable timely responses.

In this scenario, the EV is simulated under multiple risk thresholds $\epsilon$. The resulting trajectories are visualized within a single figure, where each trajectory is distinguished by a unique outline color corresponding to a specific $\epsilon$ value. This setup facilitates a direct comparison of how varying risk sensitivity affects the timing and position of the  lane change maneuver under identical environmental conditions.

Additionally, a relatively large cumulative probability threshold ($\delta = 0.2$) is adopted, such that for each SV only the most probable maneuver (longitudinal cruising in the current lane) is retained. 

\subsubsection{Key Parameters}

The simulation and vehicle parameters are summarized in Table~\ref{tab:case3_parameters}, and the initial conditions of all agents are listed in Table~\ref{tab:init_case3}.

\begin{table}[ht]
\centering
\captionsetup{font=footnotesize, skip=6pt}
\caption{Simulation and vehicle parameters (Case~3)}
\setlength{\tabcolsep}{4pt}
\label{tab:case3_parameters}
\begin{threeparttable}
\begin{tabular}{lll}
\toprule
\textbf{Parameter} & \textbf{Description} & \textbf{Value} \\ \midrule
$T_{\mathrm{sim}}$  & Simulation duration                        & 5 s \\
$T_l$               & Low-level sampling period                  & 0.02 s \\
$T_h$               & High-level update period                   & 0.08 s \\
$d_{\mathrm{safe}}$ & Minimum longitudinal safe distance         & 6 m \\
$Q_{xy}$       & Covariance in $(x,y)$ dimensions           & diag\{0.9, 0.9\} \\
$\delta$            & Cumulative probability threshold           & 0.2 \\ \bottomrule
\end{tabular}
\end{threeparttable}
\end{table}
\begin{table}[h]
\vspace{-3mm} 
\centering
\captionsetup{font=footnotesize, skip=6pt}
\caption{Initial conditions (Case~3)}
\setlength{\tabcolsep}{4pt}
\label{tab:init_case3}
\begin{threeparttable}
\begin{tabular}{lllll}
\toprule
\textbf{Vehicle} & \textbf{Initial Lane} & $x_0$ (m) & $y_0$ (m) & $v_0$ (m/s) \\ \midrule
EV        & Lane 1  & 90  &  4  & 10 \\
SV 1 (blue)   & Lane 1  & 100 &  4  & 8 \\
SV-Right (green) & Lane 2 & 60  & 0  & 9 \\ \bottomrule
\end{tabular}
\end{threeparttable}
\end{table}

\subsubsection{Observations and Results}

Fig.~\ref{fig:case3_visual} illustrates the EV’s decision-making process during the stage of a lane-change maneuver under five different risk thresholds. Each subfigure corresponds to the system state at $t = 0$, 1, 2, 3, and 4 seconds, respectively. As shown in the figure, the choice of $\epsilon$ has a significant impact on the timing and spatial position of the lane change decision. Smaller values of $\epsilon$ prompt earlier decisions both temporally and longitudinally, reflecting a more conservative planning maneuver. Conversely, larger risk thresholds delay the maneuver initiation, as the planner tolerates a higher level of uncertainty before executing the action. A more quantitative comparison of the lane change timing and position under different risk levels is provided in Fig.~\ref{fig:case3_lanechange1}.

\begin{figure}[htbp]
   \captionsetup{font=footnotesize}
   \centering
   \includegraphics[width=0.49\textwidth]{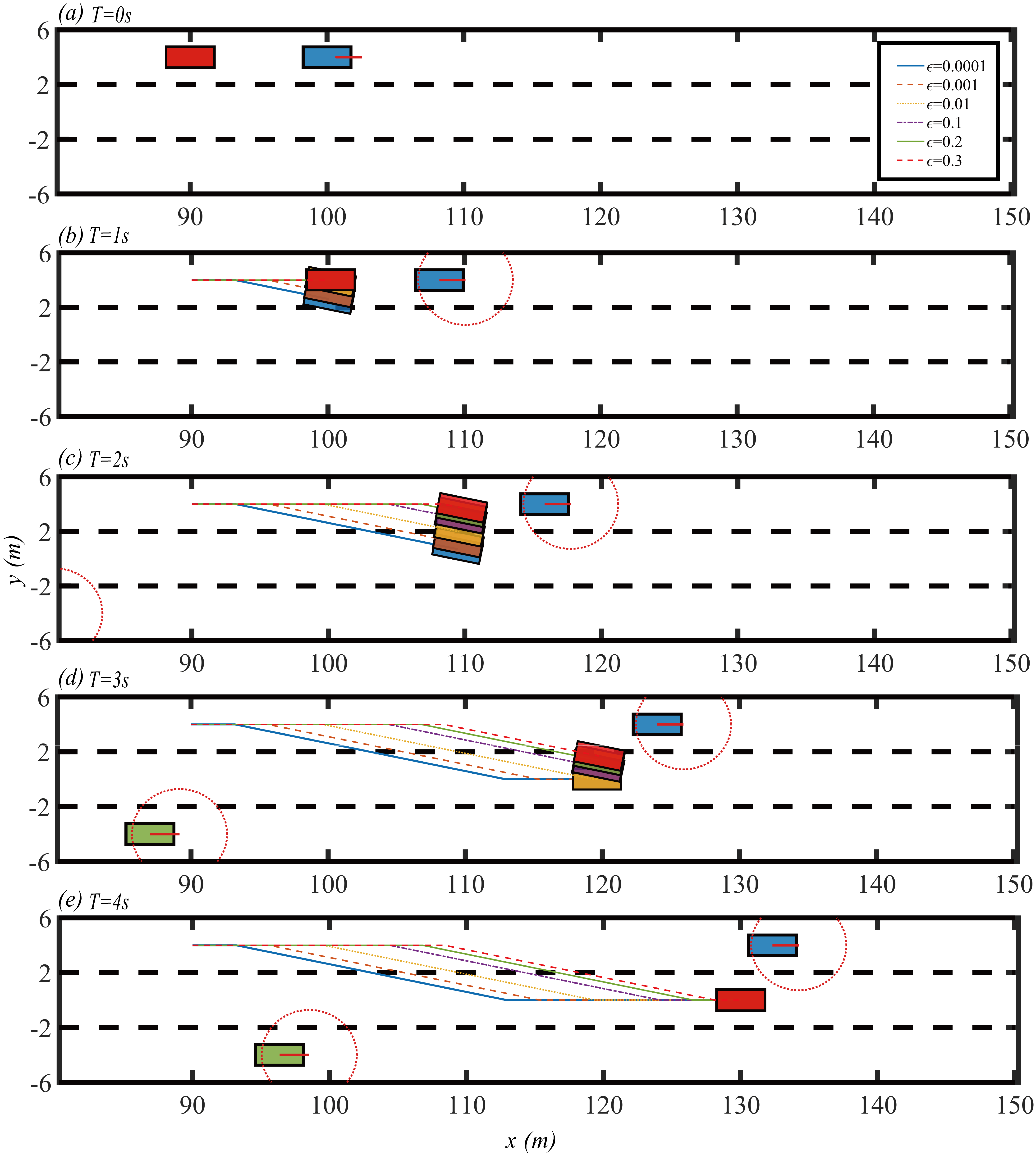}
   \caption{ Lane change decision under different risk thresholds $\epsilon$.
   Subfigures (a)-(e) correspond to $t = 0$-$4$~s, respectively.}
   \label{fig:case3_visual}
\end{figure}

\begin{figure}[htbp]
   \captionsetup{font=footnotesize}
    \centering
    \includegraphics[width=0.49\textwidth]{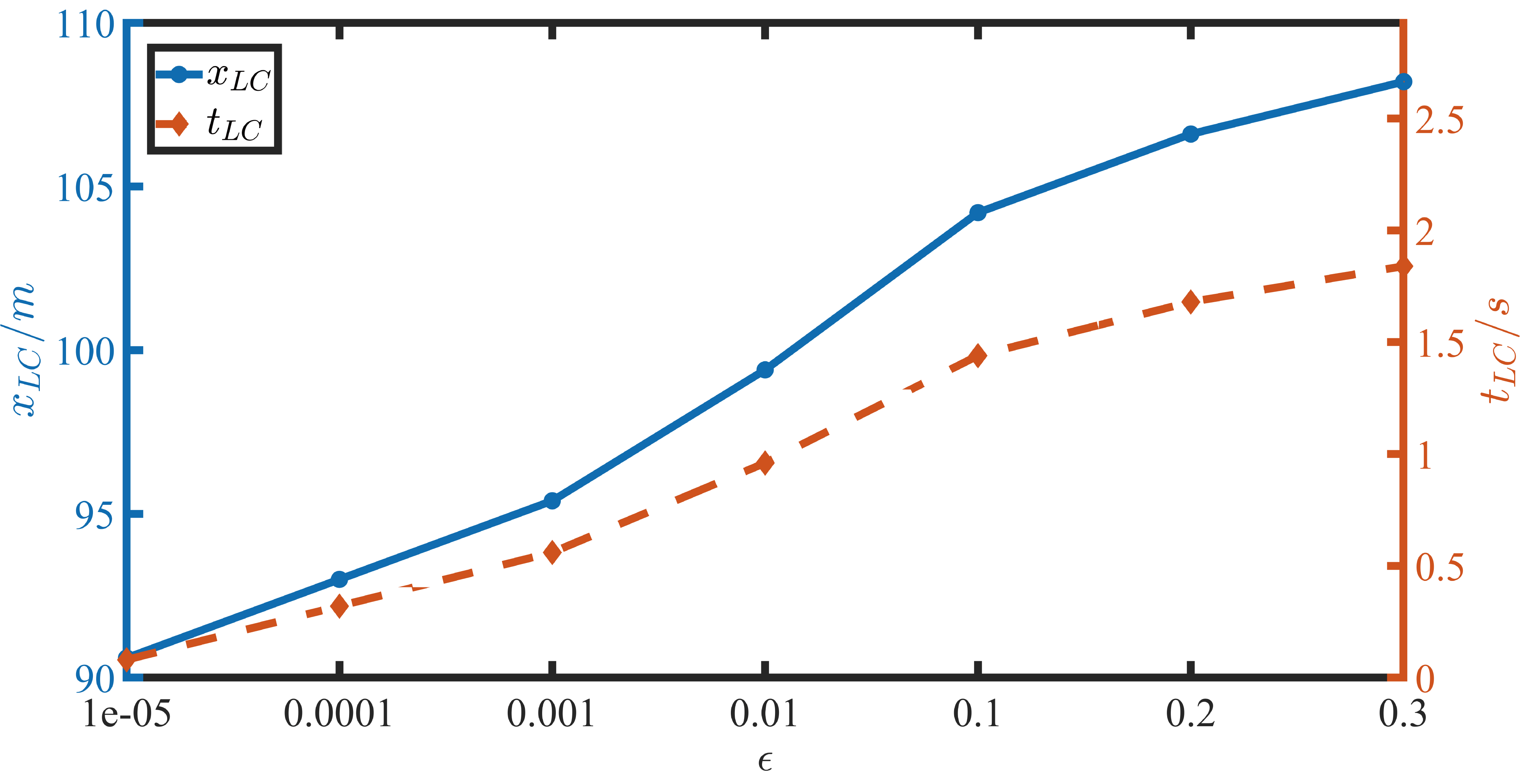}
    \caption{First lane change decision under different risk probabilities: lane change position $x_{\text{LC}}$ and decision time $t_{\text{LC}}$.}
    \label{fig:case3_lanechange1}
\end{figure}

In summary, Case 3 focuses on validating the framework’s risk-aware planning behavior under different safety confidence levels.  As the chance-constraint confidence parameter $\epsilon$ varies, the EV adapts its trajectory to trade off between conservativeness and efficiency, demonstrating that the proposed formulation enables explicit regulation of interaction risk in probabilistic terms.
\begin{remark}
    The policy of SVs are predefined and not the focus of this study. Accordingly, the cumulative probability thresholds $\delta_{\mathrm{seq}}$ is manually selected to balance scenario coverage. Small values are used in Case~1 and Case~2 to preserve all plausible behaviours, while a higher threshold is used in Case~3 to simplify analysis.

\end{remark}

Collectively, the three cases validate the proposed framework’s effectiveness across different layers of uncertainty and risk.
Case 1 and Case 2 demonstrate reliable decision-making under isolated lateral and longitudinal uncertainties, respectively, with Case 2 further highlighting the key difference in decision-making philosophy between the proposed framework and a conventional rule-based approach.
Case 3 verifies that the chance-constraint formulation can explicitly regulate probabilistic risk in the integrated optimization.
Together, these results confirm the framework’s capability for safe, adaptive, and risk-aware decision-making in complex driving environments.

\section{Conclusion}\label{V}

This paper presents a decision-making framework for autonomous systems operating under uncertainty. SAs are modeled using HMDPs, which capture both maneuver-level and dynamic-level uncertainty, thereby enabling multi-modal predictions of future behaviors. At the decision layer, an HMDP-MPC framework is developed that unifies the EA’s maneuver selection with dynamic feasibility within a single optimization process. Environment-derived multi-modal predictions are explicitly incorporated into the optimization via joint chance constraints to ensure probabilistic safety. Theoretical analysis establishes recursive feasibility and asymptotic stability of the proposed MPC scheme.

Simulation studies based on autonomous driving scenarios validate the effectiveness of the proposed framework in uncertain traffic environments and further investigate the influence of varying risk thresholds on decision-making behavior. Comparative evaluations against rule-based baselines highlight the fundamental differences in decision-making philosophy, demonstrating the proposed framework’s capability in anticipation and autonomous operation.

Future work will focus on probabilistic prediction of maneuver modes to learn more accurate distributions over MDP states, enabling more robust multi-modal reasoning and improving the reliability of decision-making in uncertain and dynamic environments.



\ifCLASSOPTIONcaptionsoff
  \newpage
\fi




\bibliographystyle{IEEEtran}
\bibliography{bibtex/bib/IEEEexample}









\end{document}